\documentclass[english]{amsart}

\usepackage{ulem}
\usepackage{graphicx}
\usepackage{amsmath,amsfonts,amsthm}
\usepackage{amssymb}
\usepackage{multirow,array}
\usepackage{comment}
\usepackage{enumitem}
\usepackage[utf8]{inputenc}
\usepackage[T1]{fontenc}
\usepackage{emptypage}
\usepackage{hyperref}
\usepackage{color}
\usepackage{mathrsfs}

\usepackage{url}

\usepackage{mathrsfs}%
\usepackage[T1]{fontenc}
\usepackage{marvosym}
\usepackage{babel}
\usepackage{amsmath,amsfonts,amsthm}%
\usepackage{amsmath}%
\usepackage{amsfonts}
\usepackage{amssymb}%
\usepackage{graphicx}

\usepackage{shuffle}

\usepackage{mathrsfs}
\usepackage{tikz}
\usepackage{mhequ}

\usepackage{geometry}
\usetikzlibrary{arrows,snakes,backgrounds}

\usetikzlibrary{snakes}
\usetikzlibrary{decorations}
\usetikzlibrary{positioning}
\usetikzlibrary{shapes}

\newcommand{\vep}{\varepsilon}

\newcommand{\tu}{\tilde{u}}

\newcommand{\MP}{\mathcal{P}}

\newcommand{\CF}{\mathcal{F}}
\newcommand{\CB}{\mathcal{B}}
\newcommand{\CK}{\mathcal{K}}
\newcommand{\CD}{\mathcal{D}}

\newcommand{\nbh}{\text{neighbourhood}}

\newcommand{\be}{\begin{equation}}
\newcommand{\ee}{\end{equation}}

\newcommand{\lan}{ \langle }
\newcommand{\ran}{ \rangle}

\newcommand{\U}{\mathcal{U}}

\newcommand{\MH}{\mathcal{H}}
\newcommand{\MF}{\mathcal{F}}

\newcommand{\MD}{\mathcal{D}}
\newcommand{\ML}{\mathcal{L}}
\newcommand{\cK}{\mathcal{K}}

\newcommand{\ka}{\kappa}
\newcommand{\eps}{\varepsilon}
\newcommand{\teta}{\tilde{\eta}}
\newcommand{\tS}{\tilde{S}}

\newcommand{\strela}{\rightharpoonup}

\newtheorem{thm}{Theorem}[section]

\newtheorem{lem}[thm]{Lemma}
\newtheorem{coro}[thm]{Corollary}
\newtheorem{rem}[thm]{Remark}
\newtheorem{prop}[thm]{Proposition}
\newtheorem{example}[thm]{Example}

\begin{document}

\title{Exponential mixing for dissipative PDEs with bounded non-degenerate noise}

\author{Sergei Kuksin}
\address{S. Kuksin, Institut de Math\'emathiques de Jussieu-Paris Rive Gauche, CNRS, Universit\'e Paris Diderot, UMR 7586, Sorbonne Paris Cit\'e, F-75013, Paris, France; and
School of Mathematics, Shandong University, Shanda Nanlu 27, 250100, PRC; and
Saint Petersburg State University, Universitetskaya nab. 7/9, St. Petersburg, Russia,}
\email{sergei.kuksin@imj-prg.fr}
 
\author{Huilin Zhang}
\address{H. Zhang, Institute of Mathematics, Fudan University, Handan Road 220, 200433, PRC,}
\email{huilinzhang2014@gmail.com}

\begin{abstract}
We prove that well posed 
quasilinear equations of parabolic type, perturbed by bounded nondegenerate random forces, 
 are exponentially mixing for a large class of  random forces. 
\end{abstract}




\maketitle

\tableofcontents

%
%

%

\newcommand{\E}{\mathbb{E}}
\newcommand{\R}{\mathbb{R}}
\newcommand{\N}{\mathbb{N}}
\newcommand{\PP}{\mathbb{P}}
 
\newcommand{\C}{\mathbb{C}}
\newcommand{\X}{\mathbb{X}}
\newcommand{\T}{\mathbb{T}}
\newcommand{\BX}{\mathbf{X}}
\newcommand{\K}{\mathcal{K}}
\newcommand{\tK}{\tilde{\mathcal{K}}}
\newcommand{\oc}{\mathcal {C}}
\newcommand{\FC}{\mathscr{C}}
\newcommand{\B}{\mathbb{B}}
\newcommand{\BB}{\mathbf{B}}
\newcommand{\op}{\mathcal{P}}
\newcommand{\FD}{\mathscr{D}}
\newcommand{\oq}{\mathcal{Q}}
\newcommand{\oor}{\mathcal {R}}
\newcommand{\hc}{\hat{c}}
\newcommand{\BI}{\mathbf{1}}
\newcommand{\BZ}{\mathbf{Z}}
\newcommand{\I}{\mathcal{I}}
\newcommand{\Z}{\mathbb{Z}}
\newcommand{\bx}{\mathbf{x}}
\newcommand{\bz}{\mathbf{z}}
\newcommand{\s}{\mathbb{S}}
\newcommand{\A}{\mathbb{A}}
\newcommand{\LL}{\mathbb{L}}
\newcommand{\tiop}{\tilde{\mathcal{P}}}
\newcommand{\BBI}{\mathbf{I}}
\newcommand{\BY}{\mathbf{Y}}
\newcommand{\FKF}{\mathfrak{K}_F}
\newcommand{\hu}{\hat{u}}


\setcounter{section}{-1}
\section{Introduction}

In this paper we  consider nonlinear PDEs, perturbed by random forces, which we write as 
\be\label{1}
\partial_t u_t+\nu  Lu_t +F(u_t) = \eta_t, \quad  u_t \in H.
\ee
Here $H$ is a Hilbert space of functions of $x$, $\nu\in(0, 1]$, $L$ is a self-adjoint operator with  compact resolvent,  $F$
is a nonlinearity and  $\eta_t = \eta^\omega_t$ is a random process in $H$. We are concerned with the question when 
a solution $u_t$ of \eqref{1} is a mixing random process, i.e. when its distribution $\MD(u_t)$ converges to a unique measure 
in $H$, independent from the initial data $u_0$, while $t\to\infty$. 
The problem of mixing in equations 
\eqref{1} is well motivated by modern physics,
and it has received much attention during the last two decades, see in \cite{KS}. In the corresponding papers (discussed in
\cite{KS}) the authors prove the mixing for various classes of equations \eqref{1}, assuming that the random force $\eta_t$ has 
 the structure which we will now discuss.  
 
 We  suppose that
 $$
 \eta_t = \sum_{i=1}^\infty b_i \eta^i_t\phi_i, \qquad \sum b_i^2<\infty,
 $$
 where $\{\phi_i\}_i$ is an  orthonormal  basis of $H$ and   $\{\eta^i_t, i\ge 1\}$ are i.i.d. real processes,  distributed as a certain etalon process $\eta^0_t$.
 Concerning the latter it was assumed that  either this is
 
 a) a kick-process, 
 $\eta^0_t = \sum_{k=1}^\infty \delta(t-kT)\xi_k,$
 where $T>0$ and 
  $\{\xi_k\}$ are i.i.d. real random variables; \ or that 
 
 b) $\eta^0_t$ is a white noise; or
 
   c)  $\eta^0_t$ is a compound Poisson process (see in \cite{KS}). 
 
 Concerning the coefficients $b_i$ it was usually assumed that either all of them are non-zero, or that  $b_i\ne0$ for all $i\le N_\nu$,
 where $N_\nu$ grows to infinity as $\nu\to0$. In the paper \cite{HM} the mixing was established  for the case when \eqref{1} is the 
 2D Navier-Stokes system on the two-dimensional torus, perturbed by a white in time random force 
  (see below eq.~\eqref{NS}), and only a $\nu$--independent 
  finite system of these coefficients do not vanish. 
  The proof of \cite{HM} uses an infinite--dimensional version of the Malliavin calculus, and all
 attempts to generalize it to equations with kick-forces or compound Poisson processes have   failed.
 
  Instead, in \cite{KNS} equations of the form \eqref{1} were considered, where $\eta^0(t)$ is a random Haar series ``of time-width one.''\footnote{ the time-width 
 one may be replaced by any  positive width.} This means the following:
 
   The processes
 $
 \eta^0\!\mid_{[k-1, k)}$, $k=1,2,\dots,
 $
 are i.i.d., so it suffices to define the process $ \eta^0\!\mid_{[0,1]}$. The latter is a random Haar series 
 $$
 \eta^0_t = \sum_{j=0}^\infty  c_j  \sum_{l=0}^{2^j -1} \xi_{jl} h_{jl}(t) \quad \forall\, 0\le t\le1,
  \qquad c_j \ne0\;\; \forall\, j, 
 $$
 where $\{ h_{jl} \}$ is the Haar base in $L_2[0,1]$ (see in \cite{Lam} and see \eqref{haar} below), and $\xi_{jl} $ are independent 
 random variables, $|\xi_{jl} |\le1$, whose density functions $ \rho_{jl} (x)$ are Lipschitz-continuous and do not vanish at $x=0$. 
 Let us denote $E= L_2([0,1],H)$ and consider the mapping 
 \be\label{S_oper}
 S: H\times E \to H, \quad S(u_0, \eta) = u_1,
 \ee
 where $u_t$ is a solution of \eqref{1}, equal $u_0$ at $t=0$.

 Assuming that 
 
 {\it 
 (B0) $\sum_j 2^{j/2} |c_j| <\infty$, so  the process $ \eta_t $ is bounded uniformly in $t$ and $\omega$ (see below in 
 Section~\ref{s_noise});

 (B1) the mapping $S$ is well defined. Moreover, there exists a compactly embedded Banach subspace
  $ V\subset H$ such that $S(H\times E ) \subset V$, the mapping 
 $S: H\times E \to V$ is analytic  and its derivatives  are bounded on bounded sets;

  (B2)  $0\in H$ is an asymptotically stable equilibrium for eq.  \eqref{1} with $\eta=0$, and relation \eqref{linearbound} 
 holds;

 (B3) $b_j \ne0$ for $j\in J\subset\mathbf N$, where the set $J$, finite or infinite, is such that  the linearised equation 
 $$
 \partial_t v+\nu Lv + dF(u_t) v = \sum_{j\in J} b_j  \xi^j_t \phi_j, \quad 0\le t\le1, 
 $$
 is approximately controllable by controls $\xi^j_t$, $j\in J$, provided that $u_t$ is a solution of \eqref{1} 
 with arbitrary $u_0$ and with $\eta=\eta^\omega$, where $\omega$ does not belong to a certain null-set depending on $u_0$, }
 \smallskip
 
 \noindent it was proved in \cite{KNS} that eq. \eqref{1} is exponentially 
  mixing. \footnote{ our statement of this result is a bit imprecise and  less general than the result of \cite{KNS} is; 
   see the  original paper for the exact statement.
 }
 
 The condition (B3) has been  intensively studied, and for many important equations it is now established that (B3) 
 holds if $J$ is a finite set, satisfying certain  explicit conditions,  independent from $\nu$. See in \cite{KNS}.
 \medskip 
 
 Our work continues the research in \cite{KNS}. Namely, making the assumptions (B$0'$) --  (B$3'$), where

 {\it (B$0'$) $\; \sum_j c_j^2 2^j <\infty$;

 (B$1'$) the mapping  $S: H\times E \to V$ is $C^2$--smooth and its derivatives up to second order are bounded on bounded sets;
 
  (B$2'$) = (B2);

  (B$3'$)  for any $u\in H$ 	and $ \eta\in E$ 
 \be\label{nondeg}
 \text{
 the operator $d_\eta S(u,\eta): E \to H$ has dense image},
 \ee
   }
we prove 
in Theorem \ref{t_final}, Section 1,  the following result:
  \medskip
  
  \noindent
  {\it there exists a unique Borel measure $\mu_{\nu}$ in $H$ such that if $u_t$ is a solution of \eqref{1} with the initial data 
  $u_0^\omega$, satisfying 
  $\ 
  u^\omega_0 \in \{u\in H: \| u\| \le R\}$ almost surely for some $R>0$, 
  then 
  \be\label{result}
  \| \MD u_k - \mu_{\nu}\|_L^* \le C \kappa^k, \quad \forall\, k\in \N. 
  \ee
  Here $\| \cdot \|_L^*$ is the Lipschitz-dual distance in the space of Borel 
  measures in $H$  (see \eqref{f_lip}),  and $C=C(R) >0$,
  $\kappa =\kappa(R) \in (0, 1)$. Moreover, if $u_0^\omega$ is any random variable in $H$, then $ \MD u_k$ 
  weakly converges to $\mu_\nu$ as $k\to\infty$. 
  }
   \medskip
   
 A special case of systems \eqref{S_oper} appears when $E$ is a subspace of $H$ and the mapping $S$ has the form
 $$
 S(u_0, \eta) = S_1(u_0) +\eta. 
 $$
 Such systems correspond to equations \eqref{1} with kick-forces $\eta$ as in a) above (see Remark~\ref{e_kicks} in 
 Sectuion~\ref{s_2dNS}). They are easier than general systems \eqref{S_oper} and for them the mixing can be proved
 if the map $S_1$ is Lipschitz--continuous on bounded sets, see \cite{KS}. 
 \medskip
   
   In Section 2 we show that the result applies to the 2D Navier-Stokes system on torus as well as to equations \eqref{1} 
   who are random perturbations of quasilinear parabolic systems which are\\
   -- well posed for any bounded 
    force $\eta(t,x)$, sufficiently smooth is $x$;\\
   -- satisfy the dissipativity assumption (B2). 
   
   The proof uses some ideas from \cite{KS, KNS} and  is significantly shorter than that in  \cite{KNS} since
   the nondegeneracy assumption (B$3'$)  which we assume now is stronger than the assumption (B3); essentially
   it holds if the force $\eta_t$ is non-degenerate (all coefficients $b_i\ne0$). 
    As    in   \cite{KNS}, the  exponential convergence \eqref{result} 
   follows from Doeblin's coupling, 
   enhanced with the quadratic convergence in the form, close  to that in the KAM--theory.\footnote{The 
   peculiarity of the used KAM--scheme is that now, in difference with the "traditional" KAM, the rate of 
   convergence is not super-exponential, but only exponential. Cf. \cite{KNS}, Section~1.3.
   }
   In particular, our work shows that the analyticity assumption in (B1)  is not an 
   intrinsic feature of the approach of  \cite{KNS}, 
   but is  needed to work under the very weak 
   nondegeneracy assumption (B3) and may be replaced by the $C^2$--smoothness if the equation  is non-degenerate
   in the sense (B$3'$). Our result is easy to apply since it is easy to check the assumption  (B$3'$). 
   \medskip

\noindent
{\it  Notation}.
As usual, by $C, C_1 ,\dots$  we denote various constants which change from line to line.  By
$ B_F(R)$ we denote the closed  $R$-ball in a Banach space $F$, centered at the origin;
by $\CD(\xi)$ --  a law of a r.v.  $\xi$;
$\langle\mu, f\rangle = \langle f,\mu\rangle $ stands for  the integral of a measurable function $f$ against a measure $\mu$,
and $\ML(X,Y)$ -- for the space of bounded linear operators between  Banach spaces $X$ and $Y$.
A complete separable metric 
space (i.e. a Polish space) $X$  always is equipped with its Borel $\sigma$-algebra $ \CB_X$.
By $\MP(X)$ we denote the space of probability measures on $(X, \CB_X)$. We  provide it with the
Lipschitz-dual distance
\be\label{f_lip}
\| \mu-\nu\|_L^* = \| \mu-\nu\|_{L(X)}^* =
 \sup| \lan \mu-\nu, f\ran|,
\ee
where the supremum is taken over all Lipschitz functions $f$ on $X$ such that
$$
|f|_{L(X)} := \sup_{x\in X} |f(x)|\, \vee \, \text{Lip}(f) \le1,
$$
with the Kantorovich distance
\be\label{Kan_d}
\| \mu-\nu\|_K = 
 \sup_{ \text{Lip}(f) \le1} | \lan \mu-\nu, f\ran|,
\ee
and the distance of total variation.
If $X,Y$ are Polish spaces and $f :X\to Y$ is a Borel--measurable mapping, then $f_*: \MP(X) \to \MP(Y)$ denotes the
corresponding push-forward of measures. A pair of random variables $\xi, \zeta$ defined on a probability space 
is \textit{a coupling} for given measures $\mu,\nu \in \MP(X)$ if $\CD(\xi)=\mu, \CD(\zeta)=\nu.$

\bigskip

{\bf Acknowledgements:} S.K. thanks the Russian Science Foundation for support through the project 18-11-00032.
H.Z. is supported by the National Postdoctoral Program for Innovative Talents, No:~BX20180075 and China Postdoctoral Science Foundation.

\section{Mixing for a class of  Markov chains}

\subsection{Settings and assumptions}
Let
$H,\ E$ be two separable Hilbert spaces and $V$ be a Banach space, compactly and densely
embedded in $H.$  Let  $S: H\times E \rightarrow V$ be a continuous mapping. We consider the following random dynamical  system (RDS) in the space $H$:
\begin{equation}\label{RDS1}
u_k=S(u_{k-1},\eta_k), \ \ k\geq 1,
\end{equation}
where $\{\eta_k \}_k$ is a sequence of i.i.d. random variables in $E$. We denote by $\ell$ the law of $\eta_k$,
$$
\ell = \MD (\eta_k) \in \MP(E)
\qquad \forall\, k,
$$
and suppose  that  $\ell$ is supported by a compact set $\K \subset E$,
$$
\K:= \text{supp}\,\ell \Subset E.
$$
By $(u_k(v), k\ge0)$, we denote a trajectory of
\eqref{RDS1} such that $u_0(v) =v$. The process $\{ u_k(v)\}$ is a Markov chain in $H$, whose transition probability function after $k$ steps is
$ P_k(u,\Gamma) = \PP(u_k(u)\in\Gamma)$. It defines a semigroup of Markov operators in the space of functions
$$
\op_k : C_b(H) \to C_b(H), \quad f(\cdot) \mapsto \E\big( f(u_k(\cdot)\big), \quad k \ge0,
$$
and a semigroup of operators  in the space of measures
$$
\op_k^*: \op(H) \rightarrow \op(H),\quad \mu\mapsto \CD(u_k(v)),\quad k \ge0,
$$
where $v$ is a r.v. in $H$, independent from the noise $\eta$, such that $\CD(v)=\mu$ (e.g. see in \cite{KS}).
\medskip

We make the following assumptions concerning  regularity of our system:

\begin{itemize}
\item[(A1)](\textbf{Regularity}). {\it The
mapping $S:H \times E \rightarrow V$ is twice continuously differentiable, and its derivatives up to second order are bounded on bounded sets. \\ }

\end{itemize}
\medskip

Concerning the  noise $\eta$ we assume the following:
\begin{itemize}
\item[(H1)](\textbf{Decomposability and non-degeneracy}). 
{\it There exists an orthonormal basis $\{e_j, j\ge1\}$ of $E$ such that
    $$
    \eta_k= \sum_{j=1}^{\infty}b_j \xi_{jk} e_j, \quad b_j \ne0 \; \forall\,j. 
$$
Here $\xi_{jk}$ are independent random variables and $b_j$ are real  numbers, satisfying
\begin{equation}\label{5}
|\xi_{jk}| \leq 1 \ \ a.s., \ \ \ \ \
\MD(\xi_{jk})=\rho_j(r)dr,
\ \ \ \
R_\eta^2:= \sum_{j=1}^\infty b_j^2 < \infty,
\end{equation}
where $\rho_j: \R \rightarrow \R $ are Lipschitz functions  and $\rho_j(0) \ne 0$ for all $j$}.

\item[(H2)](\textbf{Dissipativity}). 
{\it  For any $u\in H$, $\eta\in \K$ we have  \begin{equation}\label{linearbound}
    \|S(u,\eta )\|_{H} \leq \gamma \|u \|_H + \beta, \qquad \|S(u,0 )\|_{H} \leq \gamma \|u \|_H  \ \text{with  some} \;0<\gamma<1, \beta>0.
    \end{equation}
        }

%

\item[(H3)](\textbf{Non-degeneracy}). {\it For any $u\in  H$ and  $\eta \in \K$ the image of the 
operator $D_\eta S(u,\eta ): E \rightarrow V$ is dense  in {$H$.}}
\end{itemize}

Note that by \eqref{5} supp$\,\ell =\CK$ belongs to the Hilbert brick
\be\label{brick}
\{ \eta = \sum\eta_ke_k: |\eta_k|\le b_k\} =: \tilde \CK,
\ee
so indeed $\CK$ is a compact set, and
$
\CK \subset B_E(R_\eta).
$

It is easy to see that since  $\rho_j(0) \ne 0$ for all $j$, then 
according to relations  \eqref{linearbound} and assumption (A1), for any $a>\gamma$ and $\delta>0,$
\begin{equation}\label{contract1}
\PP\left( \|S(u,\eta)\|_H < a\|u\|_H  \right)> p_\delta > 0, \ \ \text{ if } \|u\|_H \geq \delta.
\end{equation}

\begin{rem}\label{r_1}
It is not necessary to define the map $S$ on the whole space $H\times E$: it suffices that it is defined on a \nbh \  $Q$ of $H\times \K$
which contains each point  $(u_0\times\eta)\in B_H(R)  \times \K$ with its vicinity in $H\times E$ 
of a positive radius $r(R)$. The map $S$ should be $C^2$--smooth and its $C^2$--norm should 
depend only on $\|u\|_H$.
\end{rem}

\begin{rem}\label{r_new_norm}
Assume that the system possesses the following additional property: there exists a function $R_+(R)$ such that if $\|u_0\|_H\le R$ 
and $\eta_1, \eta_2,\dots$ are arbitrary points in $\K$, then $\|u_k\|_H\le R_+(R)$ for all $k\ge0$. Then in order to 
study solutions of \eqref{RDS1} with $\|u_0\|_H\le R$  we can do the following:

-- define $O\subset H$ as a union of all trajectories of  \eqref{RDS1} with $\|u_0\|_H\le R$ and $\eta_k\in\K$ for all $k$;

-- consider the closure $\bar O$. This is a closed subset of $ B_H(R_+)$, invariant for  \eqref{RDS1} a.s. 

Then to study solutions with $\|u_0\|_H\le R$ we can work with system's restriction to $\bar O$. In particular it suffices to verify assumptions
(A1) and  (H3) for $u\in B_H(R_+)$.  Even more, it suffices 
to check (A1) and (H3) for $u\in B_H(R_+)$ with the norm $\| \cdot \|_H$ replaced by any equivalent Hilbert norm $\| \cdot \|'_H$
 in the space $H$, depending on $R_+$ (i.e., depending on $R$).  Applying  Remark~\ref{r_1} 
 we observe that  it suffices to verify  (A1) and (H3) on
 the set $Q \cap (B_H(R_+ )\times E)$. 
\end{rem}

\subsection{Main results}

Here we formulate  the  main results of this paper.
In what follows $\eta$ stays for  an element of $\K$ or for a random variable with the law  $\ell$,
 depending on the context. We will use a modified distance in the space $H$  which depends on a parameter
 $d_0\in (0,1]$:
 $$
  \|\xi_1- \xi_2\|_{d_0}:=\|\xi_1- \xi_2\|_H \wedge d_0.
 $$
 Clearly $ \|\xi_1- \xi_2\|_{d_0} \le  \|\xi_1- \xi_2\|_H$ and
 \be\label{d_dist}
 \|\xi_1- \xi_2\|_H \le \frac{2R_*}{d_0} \|\xi_1- \xi_2\|_{d_0}  \quad \text{if} \quad \xi_1, \xi_2\in B_H(R_*), \text{with } 2R_* \geq d_0.
 \ee

\begin{thm}\label{mainthm}
Under the assumptions (A1), (H1), (H2) and (H3),  
   for any $R_*>0$ and any
$u,u' \in   B_{H}(  R_*),$
\begin{equation}\label{statement}
\| P_k(u,\cdot)- P_k(u',\cdot) \|_L^* \leq C \| u-u' \|_H \, \ka^k,
\end{equation}
where $0<\ka<1$  and $C>0$ depend on $R_*$.
\end{thm}

\begin{coro}\label{c_1}
If $\mu, \nu\in \MP(B_H(R_*))$,   then
\be\label{contruction}
\|\MP_k^* \mu - \MP_k^* \nu \|_L^* \le  C_1
\| \mu -  \nu \|_L^*\,  \kappa^k\quad \forall\, k\ge1,
\ee
where $0<\ka<1$  and $C_1>0$ depend on $R_*$.
\end{coro}
\begin{proof}
Let $(u^{\omega'}, u'^{\omega'})$ be a coupling for $(\mu, \nu)$. Then
$
\MP_k^*(\mu) =\CD u_k(u)$, $ \MP_k^*(\nu) =\CD u_k(u').
$
Take any $f\in C_b(H)$ with $|f|_{L(H)}\le1$. Then
\begin{equation*}
\begin{split}
| \lan f, \MP_k^*(\mu) -  \lan f, \MP_k^*(\nu) \ran| =\big| \E^{\omega'} \E^{\omega} (f(u_k(u^{\omega'}))-   f(u_k(u'^{\omega'})))\big|
\le C\kappa^k  \E^{\omega'} \| u^{\omega'} - u'^{\omega'} \|_H,
\end{split}
\end{equation*}
where the inequality follows from the theorem. This inequality remains true  if in the r.h.s. we take infimum over all
couplings  $(u, u')$ for $(\mu, \nu)$.  By the Kantorovich--Rubinshtein theorem
(see in \cite{KS})
$$
\inf_{u,u'} \E^{\omega'} \| u^{\omega'} - u'^{\omega'} \|_H = \| \mu -\nu\|_K,
$$
where $\|\cdot \|_K$ is the Kantorovich distance \eqref{Kan_d}.  Since $\nu, \mu$ supported on $ B_H(R_*)$, then
$\|\mu -\nu\|_K \le R_* \|\mu-\nu\|_L^*$, and the result follows.
\end{proof}

Clearly without loss of generality we may assume that in Theorem \ref{mainthm}
\begin{equation}\label{R_*}
R_*\ge \frac{\beta}{1-\gamma} =:  R_*^0\,;
\end{equation}
this relation is always assumed from now on. Then in view of (H2)
$$
S(u,\eta) \in B_H(R_*) \quad\text{if} \quad u\in B_H(R_*), \eta\in \K.
$$
By $(A1)$ if $u\in B_H(R_*)$ and $\eta \in \K$, then  $\| S(u,\eta)\|_V \le K(\K, R_*)$. We see that
the set
$$
X= X_{\K, R_*} = \text{completion in $H$ of $B_H(R_*) \cap B_V(K(\K, R_*))$},
$$
is a compact subset of $H$ such that
$$
S:  B_H(R_*)
\times \K \mapsto X.
$$
In particular,
\be\label{hihi}
S: X\times \K \mapsto X.
\ee
That is, the RDS \eqref{RDS1} defines a Markov chain in  $X$ and
$$
S(u, \eta) \in X \quad \text{a.s. \ \ \ if}\;\; u\in B_H(R_*) .
$$
From here for $u,u'$ as in the theorem's assumption we have
$$
u_k(u),\; u_k(u') \in X\quad\text{for\; $ k\ge1$,\; \; a.s.}
$$
The proof of Corollary \ref{c_1} and this relation show that \eqref{contruction} holds if \eqref{statement} is established only for $u, u'\in X$. But 
\eqref{contruction} implies the validity of the assertion of Theorem~\ref{mainthm} for all $u,u'\in B_H(R_*)$. So
\begin{equation}\label{attention}
\begin{split}
&\text{ \it proving the theorem we may assume that $u,u' \in X$ and }\\
&\text{ \it   regard \eqref{RDS1} as a system in the compact set $X$.}
\end{split}
\end{equation} 
\medskip

Choosing in \eqref{contruction} $k$ sufficiently big we find that the operator
$\MP_k^*$ defines a contraction  of the complete
metric space $(\MP(X), \| \cdot\|_L^*)$. So it has a unique fixed point $\mu_*\in\MP(X)$. Then
$\MP_k^* \MP_1^* \mu_* =  \MP_1^* \mu_*$,  so $ \MP_1^* \mu_*=\mu_*$ by the uniqueness, i.e. $\mu_*$ is
a stationary measure for the Markov chain, defined by \eqref{RDS1} on $X$. Due to  \eqref{contruction}  it is unique. Since
$X=X_{\K, R_*}$, then formally this measure depends on $R_*$, $\mu_* = \mu_*(R_*)$. But the measure
$ \mu_*(R^0_*)$ (see \eqref{R_*})
is stationary for the system on $X_{\K, R_*}$ for any $R_*\ge R_*^0$,  so $ \mu_*(R^0_*)=  \mu_*(R_*)$
by the uniqueness. Denoting this measure by $ \mu_*$ we see that it is stationary for the system \eqref{RDS1}, considered on
$H$, and derive from \eqref{contruction} (with a suitable $R_*$) that
\be\label{LL}
\| P_k(u, \cdot) -\mu_*\|_L^* \to 0 \quad\text{as}\quad k\to\infty,
\ee
for every $u\in H$.
Assume that $\mu'$ is another stationary measure for the system on $H$. Then for any bounded Lipschitz function $f$ on $H$
we have
$$
\lan f, \mu'\ran = \lan f, \MP_k^* \mu' \ran = \lan \MP_k f, \mu'\ran = \int_H \lan f, P_k(u, \cdot )\ran \, \mu'(du) \to \lan f,\mu_*\ran
$$
by the Lebesgue theorem and \eqref{LL}. So $\mu' = \mu_*$. We finalize our analysis of eq.~\eqref{RDS1} in the following
theorem:

\begin{thm}\label{t_final}
Under the assumptions (A1), (H1), (H2) and (H3) equation \eqref{RDS1} defines in $H$ a Markov chain which has a unique
stationary measure $\mu_*$. This measure is supported by the ball $B_H(R_*^0)$, and if $\mu\in \MP(H)$ is
supported by a	 ball $B_H(R_*)$, $ R_*\ge R_*^0$, then
\be\label{conv}
\| \MP_k^*\mu - \mu_* \|_L^* \le C\kappa^k, \quad \forall\, k\ge1,
\ee
where $C>0$ and $\kappa\in(0,1)$ depend on $R_*$. If $\nu$ is any measure in  $ \MP(H), $  then 
\be\label{last_stat}
 \MP_k^* \nu \strela \mu_*\quad  \ \text{as }k \rightarrow \infty,
\ee
where $\strela$ signifies the weak convergence of mesures.

\end{thm}
\begin{proof}

The first part of the theorem follows from what was said above if we note that the measure $\MP_1^*\mu$ is supported by the compact set
$
X_{\K, R_*\vee R_*^0},
$
so \eqref{conv} is a consequence of   \eqref{contruction} with $\mu := \MP_1^*\mu$ and $\nu :=\mu_*$. Now consider any  $\nu\in\MP(H).$ For $R$ large, define a probability measure supported by $B_H(R)$: 
 $$
 \nu_R(\cdot):= \frac{1}{\nu(B_{H}(R))} \nu(\cdot \cap B_H(R)).
 $$
  In view of  Ulam's theorem  $\nu_R $ converges to $\nu$ in  the total variation distance. We have 
\begin{eqnarray*}
\|\MP_k^* \nu- \mu_* \|_L^* &\leq &  \|\MP_k^* \nu- \MP_k^*\nu_R \|_L^* +  \|\MP_k^* \nu_R- \mu_* \|_L^*\\
& = & \sup_{ | f|_{L(H)}\le1 }\ \lan \MP_k f,   \nu-  \nu_R \ran + C_R \kappa_R^k 
\, \le \,  2 \| \nu-\nu_R  \|_{var} + C_R \kappa_R^k.
\end{eqnarray*}
 Choosing first $R$  so big that the first term in the r.h.s. 
  goes to zero and then $k$ so large such that the second term goes to null we see that the l.h.s. converges to zero with $k$.
This implies \eqref{last_stat}. 
\end{proof}

Note that since the theorem deals with initial data, supported by a ball $B_H(R_*)$, then Remark~\ref{r_new_norm} applies and
it suffices to check Assumption~(H2) in the weaker form, specified in that remark. 
\medskip

%
%

\subsection{Main lemma and proof of Theorem \ref{mainthm}}
In this section we state the  main technical  lemma and derive from it Theorem \ref{mainthm}.
Based on  \eqref{attention} we regard \eqref{RDS1} as a system on the compact set $X=X_{\K, R_*}$.

\begin{lem}\label{mainlemma}
Under the assumptions (A1), (H1), (H2), (H3), there exist constants $C>0$, $\delta\in(0,1]$ and a continuous
mapping $\Phi: X \times H \times \K \rightarrow E$ such that $\Phi(u,u';\eta)=0$ if $u=u'$, and  the mapping
$$
\Psi(u,u';\eta):=\eta + \Phi(u,u';\eta)
$$
satisfies
\begin{eqnarray}
\|\ell - \Psi_*(u,u';\cdot)\ell\|_{var } &\leq& C \|u-u'\|_H, \label{3.8}\\
\|S(u,\eta)-S(u',\Psi(u,u';\eta))\|_H &\leq& \tfrac12 \|u-u'\|_H,  \label{3.9} 
\end{eqnarray}
 for all $u,u' \in D_\delta :=\{(u,u')\in X \times H : \|u-u'\|_H \leq \delta \}$ and all $\eta\in \K$.
\end{lem}

 Lemma \ref{mainlemma} will be proved in the next section. Now, supposing that  we have this result,  we will first establish
 a  coupling lemma and then derive from it   Theorem~\ref{mainthm}.

\begin{lem}\label{coupling}
For $\delta$ as in Lemma \ref{mainlemma}
 there exists $C_1=C_1(\delta)$, a probability space  $(\Omega_0, \MF_0, P_0)$
and a measurable mapping
$$
D_\delta   \times \Omega_0 \to \CK\times \CK, \quad
(u, u', {\omega_0}) \mapsto (\eta_1^{\omega_0}(u, u'), {\eta'_1}^{{\omega_0}}(u,u'))
$$
such that $\CD(\eta_1(u,u'))=\CD(\eta'_1(u,u'))=\ell$, and
$u_1:=S(u,\eta_1)$, $u'_1:=S(u', \eta'_1)$ satisfy
\begin{equation}
P\big(\|u_1-u'_1\|_H\leq \tfrac12 d\big)> 1 - C_1 d, \quad d=\|u-u'\|_H \le\delta. 
\end{equation}
\end{lem}

\begin{proof}
Recall that $\eta$ has the law  $\ell.$  Denote $\eta'=\Psi(u,u',\eta)$ and $ \ell' =\CD \eta'$.
Then according to Lemma~\ref{mainlemma},
\begin{eqnarray}\label{3}
&&\|S(u, \eta)-S( u',\eta')\|_H \leq \tfrac12 d, \ a.s.\\
&&\|\ell- \ell' \|_{var} \leq C_1 d.\label{3.1}
\end{eqnarray}
The law $\ell'$ of $\eta'$ need not to be $\ell$. To improve this, note that in view of \eqref{3.9} and the
 Dobrushin lemma (see in \cite{KS}), there exists a coupling $(\teta,\teta')$ for $(\ell, \ell')$, 
  defined on a probability space $(\Omega_1,\MF_1,P_1)$, satisfying
\begin{equation}\label{44}
P_1(\teta\neq \teta')= \| \ell - \ell'  \|_{var} \leq C_1 d.
\end{equation}
Moreover, the mapping $(\omega_1, u, u')\mapsto (\tilde\eta, \tilde\eta')$ may be chosen to be measurable, see \cite{KS},
Section~1.2.4. For the pair of measures $(\ell, \ell')$ we have obtained two couplings -- $(\eta, \eta')$ and
$(\tilde\eta, \tilde\eta')$.

 According to the \textit{Gluing lemma} (see \cite{Vil03}), there exist random variables $(\zeta_1,\zeta_2,\zeta_3)$ defined on a probability space
 $(\Omega_0,\MF_0,P_0)$ such that
$$
\MD(\zeta_1,\zeta_2) {=} \MD(\eta,\eta') ,\ \  \ \MD(\zeta_2,\zeta_3) {=}\MD(\teta',\teta) .
$$
The triplet  $(\zeta_1, \zeta_2, \zeta_3)$ may be chosen to be a measurable function of $(\omega_0, u, u')$, see
\cite{KNS},~Appendix~3.
In particular,  $\MD(\zeta_1)=\MD(\zeta_3)=\ell $ and by inequality \eqref{3}
$$
\|S(u,\zeta_1)-S(u',\zeta_2)\|_H \leq \tfrac12 d, \ \ P_0-a.s.
$$
Furthermore, in view of  inequality \eqref{44},
$$
P_0(\zeta_2 \neq \zeta_3) =P_1(\teta'\ne \teta)
\leq C_1 d,
$$
which implies that
$$
P_0\big(\|S(u,\zeta_1)-S(u',\zeta_3)\|_H \leq \tfrac12 d\big)\geq 1-C_1 d.
$$
To complete the proof it remains  to choose $\eta_1=\zeta_1$ and  $\eta'_1=\zeta_3$.

\end{proof}

To obtain the exponential mixing, claimed by Theorem \ref{mainthm}, we will apply
the method of   Kantorovich functional (see in \cite{KS}). A  Kantorovich functional for measures on $X$
is a symmetric function
$
K: \MP(X) \times \MP(X) \to \R_+,
$
defined via its density $f_K$. The latter  is a  symmetric measurable function $f_K$ on $X\times X$, satisfying
$$
f_K(\xi_1,\xi_2) \geq \|\xi_1- \xi_2\|_{d_0} \quad \forall\, \xi_1,\xi_2\in X,
$$
 (see \eqref{d_dist}), and   $K$ is defined
in terms of  $f_K$ by the relation
$$
K(\mu_1, \mu_2) =\inf\{\E f_K(\xi_1, \xi_2)\},
$$
where the infimum is taken over all couplings $(\xi_1, \xi_2)$ for $(\mu_1, \mu_2)$.
\medskip

For $R_*$ as above and $u,u'\in B_H(R_*)$ we denote $d=\| u-u'\|_H$ and set
\be\label{new_rule}
 (u_1, u'_1)(\omega_0)
 =
   \left\{\begin{array}{ll}
 (S(u,\eta), S(u',\eta)),
& d>d_0 \,,
 \\
  (S(u,\eta_1), S(u',\eta'_1)),   & d\le d_0\,,
\end{array}\right.
\ee
where $\eta_1(\omega_0)$,  $\eta'_1(\omega_0)$  are as in Lemma \ref{coupling}, and $\eta$ is a r.v. on $\Omega_0$, distributed
as $\ell$, and independent  from $\eta_1, \eta'_1$. Clearly
$(u_1, u_1') $ is a measurable function of $(u, u', \omega_0)$ and $\CD(u_1) =  P_1(u, \cdot)$ and $\CD(u'_1) =  P_1(u', \cdot)$.

\begin{lem}\label{1stepcontract}
There exist $d_0>0$ and a function
\be\label{about_f}
f:(\tfrac12d_0, R_*] \to (2d_0, 3d_0]\quad \text{non-decreasing},
\ee
such that for any $\xi_1, \xi_2 \in X,$
\be\label{f_K}
f_K(\xi_1, \xi_2)=\left\{
\begin{array}{ll}
\|\xi_1 - \xi_2\|_H, & \text{if }\|\xi_1-\xi_2\|_H\leq d_0,\\
f(\|\xi_1\|_H\vee \|\xi_2\|_H), & \text{if }\|\xi_1-\xi_2\|_H > d_0,
\end{array}
\right.
\ee
is  a Kantorovich density. It satisfies
\begin{equation}\label{onestep}
\E\left[f_K(u_1,u'_1)\right]\leq \kappa f_K(u,u') \ \ \text{ for any }u,u' \in B_H(R_*),
\end{equation}
where $\kappa\in(0,1)$ is a constant, depending on $R_*$.
\end{lem}

\begin{proof} We may assume that   $d_0 \leq \delta,$ where $\delta$ is introduced  in Lemma \ref{mainlemma}.
The fact that $f_K$ is   a Kantorovich density is obvious, as well as that it satisfies
\be\label{relation}
f_K(\xi_1, \xi_2) \le 3 \| \xi_1- \xi_2\|_H.
\ee
It remains to verify \eqref{onestep}.
 Let us denote
 $R_1(\omega_0):=\|S(u,\eta)\|_H \vee \|S(u',\eta')\|_H$, $R:=\|u\|_H \vee\|u'\|_H $. Recall that  $d=\|u-u'\|_H$.\\
{\bf Case 1: }$d>d_0$.\\
Now $ \tfrac12 d_0 \le R\le R_*$. Let as take $a=\frac{1+\gamma}2$.
According to  \eqref{contract1} and \eqref{linearbound},
$$
P[R_1 \leq a R ] > p_{d_0}=:p>0, 
\qquad
R_1 \leq   \gamma  R+\beta ,\ \ a.s.  
$$
Therefore
$$
\E[f_K(u_1,u'_1)] \leq p f(aR) + (1-p) f(\gamma R+\beta).
$$
To obtain inequality \eqref{onestep}, we need to find  $f$ of the form \eqref{about_f}, satisfying
\begin{equation}\label{8}
p f(aR) + (1-p) f(\gamma R+\beta) \leq \kappa_1 f(R)\quad \text{for}\quad R\in (\tfrac12 d_0, R_*],
\end{equation}
with some $\kappa_1 \in (0,1).$  If $R\ge R^0_* =\beta/(1-\gamma)$, then  $\gamma R+\beta \le R$, so \eqref{8} holds with some
$\kappa_1=\kappa_1(f) <1$ if in addition to \eqref{about_f} we assume that
\be\label{ab_f}
f \quad \text{is strictly increasing on $\; [R^0_*, R_*]\ $ and $f(aR^0_*)< f(R^0_*)$.
}
\ee
For any fixed $d_0$ with $d_0<2R^0_*= {2\beta}/(1-\gamma)$, it remains to  consider $f$ on the segment   $J:=(\frac12 d_0, R^0_*]$.

Consider the segments
$$
I_1=(aR^0_*, R^0_*], \;I_2=(a^2R^0_*, a R^0_*], \dots,
I_n=(a^nR^0_*,a^{n-1}R^0_*], \dots
$$
Then
$
\cup_{j=1}^{N_0} I_j \supset  J,
$
where
$
 N_0 =\big[\frac{\ln2R_0 + \ln d_0^{-1} }{\ln a^{-1}} \big]+1.
$
We define the function $f$ on $J$ by  relation
$$
f(x)= a_j,\ \ x\in I_j, \quad  1\le j\le N_0,
$$
where
\be\label{a_j}
2.5 d_0= a_1> a_2 >\dots > a_{N_0}>0
\ee
is a sequence to be defined. If $R\in I_1$, then  \eqref{8} is valid if
$$
a_1 > p a_2 + (1-p) a_1,
$$
which holds true since $a_1>a_2$.  If $R\in I_{n}$,  $n \ge 2$, then \eqref{8} holds if
\begin{equation}\label{10.2}
a_{n-1}>p a_n +(1-p) a_1.
\end{equation}
Let us take some $p_1\in (0, p)$ and define the sequence $\{a_j\}$ by  relation
$$
a_n:=a_1- \frac{a_1-a_2}{p_1^n}.
$$
If $a_2<a_1$ is chosen sufficiently close to $a_1$, then \eqref{a_j} holds. Relation \eqref{10.2} holds as well and
 \eqref{onestep} is proved for $d>d_0$
with some $\kappa = \kappa_1<1$.

Choosing $a_1-a_2>0$ sufficiently small we achieve that $2d_0 \le f(x) \le 2.5 d_0$ on $J$. If
$R^0_*<R_*$, then we extend $f(x)$ to an affine function on $[R_0, R_*]$, equal to $3d_0$ at $x =R_*$. We have arrived at
a function $f$, satisfying \eqref{about_f} and  \eqref{ab_f}, such that  \eqref{relation} holds if $d>d_0$. 
\medskip

\noindent{\bf
Case 2.} $d\leq d_0.$\\
Denote $D=\|u_1-u'_1\|_H$. Since  $d_0 \leq \delta,$ then according to Lemma \ref{coupling},
$$
\PP\big(D\leq \tfrac12 d\big) \geq 1-C_1 d\quad\text{if}\quad d\le d_0.
$$
In view of \eqref{about_f},
$$
\E[f_K(u_1,u'_1)] \leq C_1 d (3d_0) + (1-C_1 d)(\tfrac12 d) \le \tfrac12 d+ (\tfrac52 C_1 d) d_0.
 $$
Choosing $d_0\le 1/(10C_1) $ we achieve that \eqref{onestep} holds with $\kappa= 3/4$.

Now the lemma is proved by choosing $d_0\le 1/(10C_1) \wedge {2\beta}/{(1-\gamma)}$ with $\kappa = \kappa_1\vee (3/4)$.
\end{proof}

\begin{proof}[Proof of Theorem \eqref{mainthm}]
Note that for any $f \in C_b(H)$ with $|f|_{L(H)} \le 1$ (see \eqref{f_lip}) and any measures $\mu,\nu\in\MP(H)$ we have
\begin{eqnarray*}
\Big|\int_H f(x)(\mu -\nu)(dx) \Big|= \left| \E\left[ f(\xi)-f(\zeta) \right] \right| \leq \E \|\xi-\zeta\|_H,
\end{eqnarray*}
where $(\xi,\zeta)$ is a coupling for $\mu,\nu.$ Therefore
$$
\|\mu-\nu \|_L^* \leq \inf_{(\xi,\zeta)} \E \|\xi-\zeta \|_H,
$$
where the infimum is taken over all couplings $(\xi,\zeta)$ for $(\mu,\nu).$ If $\mu, \nu$ are supported by  $X$, then the couplings are valued in $B_H(R_*)$ a.s., and due to  \eqref{f_K} and \eqref{about_f}
we have that
$$
\| \xi^\omega- \zeta^\omega\|_H \le \frac{R_*}{d_0} f_K(\xi^\omega, \eta^\omega) \quad \text{a.s.}
$$

For $u, u'\in X$ and $k\ge1$ let $(u_k, u'_k)$ be any fixed 
 coupling for $(P_k(u, \cdot), P_k(u', \cdot)$.  Then
\be\label{controllip}
\|P_k(u,\cdot)- P_k(u',\cdot)\|_L^*  \leq  \inf_{ \xi, \zeta} \E \| \xi-\zeta \|_H
\leq  \frac{R_*}{d_0} \inf_{\xi, \zeta}  \E f_K(\xi, \zeta)
\leq   \frac{R_*}{d_0}    \E f_K(u_k, u'_k),
\ee
where, as above, $(\xi, \zeta)$ is a coupling for the two measures.

 Denote $\mu_k=P_k(u,\cdot)$ and $\mu'_k=P_k(u',\cdot)$. We claim that for any $k\geq 1$ there exists a coupling $(u_k, u'_k)$
 for the measures $(\mu_k, \mu'_k)$, defined  on a probability space $(\Omega^{k}, \MF^{k}, P^{k})$,   such that
\begin{equation}\label{28}
\E^{k}\left[ f_K(u_k, u'_k) \right] \leq \kappa^k  f_K(u,u'),
\end{equation}
where $\kappa \in (0,1).$ 

For $k=1$ this is the assertion of Lemma \ref{1stepcontract}.
Now let $k\ge2$. Consider a coupling $(u_{k-1}(\omega_{k-1}), u'_{k-1}(\omega_{k-1}))$ for $(\mu_{k-1}, \mu'_{k-1})$,
defined on a probability space $(\Omega_{k-1}, \CF_{k-1}, P_{k-1})$ and satisfying \eqref{28} with ${k:=k-1}$,
existing by the base of induction. Let
$(\Omega_{0}, F_{0}, P_{0})$ be the probability space from Lemma~\ref{coupling}. We take for $\Omega_k$ the direct
product of these two spaces: $\ \Omega_k = \{\omega_k = (\omega_{k-1}, \omega_0)\}$, and set
$
(u_k, u'_k) = (u_1, u'_1)  (\omega_{k-1}, \omega_0),
$
where the pair $ (u_1, u'_1)$ is defined as in \eqref{new_rule} with $u=u_{k-1}(\omega_{k-1}), u'=u'_{k-1}(\omega_{k-1})$.

By \eqref{onestep} and the base of induction we have:
\begin{equation*}
\E f_K(u_k, u'_k) = \E^{k-1}\Big[ \E^0 f_K\big(u_1, u'_1\big)(u_{k-1}(\omega_{k-1}, u'_{k-1}(\omega_{k-1}))\Big]
\le\kappa \E^{k-1}f_K(u_{k-1}, u'_{k-1}) \le \kappa^k f_K(u,u'),
\end{equation*}
so \eqref{28} is proved.
By  \eqref{controllip}, \eqref{28} and \eqref{relation}
$$
\| P_k(u, \cdot) - P_k(u', \cdot)\|_L^* \le \frac{R_*}{d_0} f_K(u, u') \kappa^k\le
 \frac{3R_*}{d_0}  \| u- u'\|_H \kappa^k,
$$
which proves Theorem \ref{mainthm}.
\end{proof}

\subsection{Proof of the main lemma}

 We start with a result on almost-inverse linear operators, depending on a parameter, following closely Section 2.2 of \cite{KNS}. 
Let $H$, $E$ and $V$
be the  spaces as above and $X$, $\K$ be the compact sets as above. For the Hilbert base $\{e_j\}$ of $E$ as in (H1) and any 
 $M\ge1$ we 
denote by  $ E_M$ the subspace of $E$, generated by the first $M$ vectors of the base, and denote by $P_M$ the orthogonal 
 projection $E\to   E_M$.

\begin{prop}\label{rightinverse}
Let
$\
A: X \times \K \rightarrow \ML(E,H)
$
be a continuous mapping. Assume that
 for any $u\in X$ and $ \eta \in \K$ the operator
$A(u,\eta)$ has  dense image in $H.$
Then for any $\vep>0$ there exists $M_\vep \in \N$, $C(\vep)>0$
 and a continuous  mapping
$\
R_\vep: X \times \K \rightarrow \ML(H, E)
$
such that for any $u \in X, \eta \in \K,$
\begin{eqnarray}\label{inverse1}
 &&\text{Im} \left(R_\vep(u,\eta)\right) \subseteq  E_{M_\vep}, \quad  \|R_\vep(u,\eta)\|_{\ML(H,E)} \leq C(\vep);\\
 &&\|A(u,\eta) R_\vep(u,\eta)f-f\|_H \leq \vep \|f\|_V \ \ \forall f \in V.\label{inverse2}
\end{eqnarray}

\end{prop}

\begin{proof}

For any $u\in X$ and $\eta \in \K$ define $G(u,\eta):= A(u,\eta) A^*(u,\eta)$, where $A^*$ is the adjoint operator. This is a non-negative
self-adjoint operator. Its kernel equals to that of $A^*$ and is trivial since $A$ has dense image. So
$
G = G^* >0,
$
and the operator  $(G+r I)^{-1}$ is well-defined and smooth in $\{r>0\}.$ For  $M\in\N$ denote
\begin{eqnarray*}
R_r(u,\eta):= A^*(u,\eta) (G(u,\eta)+ r I)^{-1},\quad
R_{r,M}(u,\eta):= P_M \circ R_r(u,\eta).
\end{eqnarray*}
For any $r$ and $M$ the  operator $R_{r,M}$ is continuous in $u$ and $\eta$.
 Below we will choose $M=M_\vep$ and $  r=r_\vep$ in
 such a way  that $R_\vep:=R_{r,M}$ satisfies \eqref{inverse1},\eqref{inverse2}.

Since $\sup_{X \times \K} \|A(u,\eta)\|<\infty$ and $\|(G+rI)^{-1}\|\leq r^{-1},$ then
$
\sup_{X \times \K} \|R_{r,M}(u,\eta)\| \leq C r^{-1},
$
which implies \eqref{inverse1}.

Now we will  prove \eqref{inverse2}. Doing that we may assume that $f\in B_V(1)$. Since $AR_r-AR_{r,M}= A(I-P_M)R_r$, then
$$
\| AR_r f- AR_{r,M}f\|_E \le C \| (I-P_M) R_r f\|_H =: CF_M(u,\eta, f),
$$
for all $(u,\eta, f) \in X\times \CK \times B_V(1) =: \frak A$. The functions $F_M$ are continuous on $\frak A$, and
$F_M\searrow 0$ as $M\to\infty$, pointwise. Since the set $\frak A$ is compact, this convergence is uniform by Dini's
theorem. That is,
$$
\sup_{\frak A}    
 \|AR_r f- AR_{r,M} f\| \to0\quad\text{as}\quad M\to\infty,
$$
for any $r>0$.
So  to obtain inequality \eqref{inverse2} it remains to   show  that
\begin{equation}\label{45}
\sup_{\frak A} \|AR_r f -f\|  \to 0 \quad \text{as} \quad r\to0.
\end{equation}
To prove this note that since
$
AR_r = A A^*(G+rI)^{-1} = G(G+rI)^{-1},
$
then
$$
\| AR_r f-f\|_E = r\| (G+rI)^{-1} f\|_E =: \Phi_r{(u, \eta, f}).
$$
Consider the family of continuous functions
$$
\Phi_r: \frak A \to \R_+, \quad r>0.
$$
Writing the operator $r(G+rI)^{-1}$ and the vector $f$ in terms of the spectral decomposition for the positive
self-adjoint operator $G(u,\eta)$, we get that the functions $\Phi_r$
pointwise monotonically converge  to zero as $r\to0$ (see \cite{KNS}, Lemma~2.4, for details).

 Since $\frak A$ is a compact set, then evoking again  Dini's theorem \footnote{A version of the theorem for families of
functions easily follow from the usual Dini's theorem for sequences of functions.
}
 we see that the convergence is uniform.
So \eqref{45} follows and the proposition is proved.
\end{proof}

Now we are ready to prove the main lemma.

\begin{proof}[Proof of Lemma \ref{mainlemma}]

{\it Step 1: construction of the mapping $\Phi$.} 
By the regularity assumption (A1), for any $\zeta \in E$
\begin{equation}\label{taylor}
S(u', \eta +\zeta)=S(u,\eta) + D_u S(u,\eta)(u'-u) + D_\eta S(u,\eta) \zeta +O(\|u'-u\|_H ^2 + \|\zeta\|_E^2).
\end{equation}
To obtain  \eqref{3.9}, consider the relation above with $\zeta=\Phi(\eta)= \Phi(u,u';\eta)$.
If we can construct a measurable mapping $\Phi: D_\delta\times \cK  \to  E_{M_\vep}$  such that
\begin{eqnarray}
&&\|D_\eta S(u,\eta) \Phi(\eta) + (D_u S(u,\eta))(u'-u) \|_H  \leq \tfrac14 \|u'-u\|_H , \\
&&\|\Phi(\eta)\|_E  \leq C \|u'-u \|_H,
\end{eqnarray}
for all  $(u,u',\eta)$ with a uniform constant $C$,  then \eqref{3.9} should follow if $\delta\ll1$.

For a fixed $(u,u', \eta) \in   D_\delta\times \K$ consider the following equation on $\zeta:$
$$
D_\eta S(u,\eta) \zeta = -(D_u S)(u,\eta)(u'-u).
$$
By assumption $(H3)$ and Proposition \ref{rightinverse}, for any $\vep>0$ there exists  an $M=M_\vep\in \N$ and
a linear operator
$R_\vep(u,\eta): H \rightarrow  E_{M_\vep}$ whose  norm is bounded uniformly in $u, \eta\in X\times \K$,
such that for any $f\in V,$
$$
\|(D_\eta S) R_\vep f-f \|_H \leq \vep \|f\|_V.
$$
Take $f=-(D_u S)(u'-u)$. By $(A1)$ $f$ belongs to the space $V$ and $\| f\|_V \le C \|u'-u\|_H$, uniformly in
$(u, \eta) \in X\times \K$. Define
$$
\Phi:   D_\delta\times \K \to   E_{M_\vep},\quad
\Phi(u,u';\eta)= R_\vep(u,\eta) f= -R_\vep(u,\eta) (D_u S(u,\eta)) (u'-u).
$$
For 	any $\vep>0$ this is a continuous mapping, and  due to  assumption $(A1)$
\begin{eqnarray}
&& \| \Phi(u,u'; \eta)\|_E \leq C_\vep \|u'-u\|_H,\label{bd1}\\
&& \|\Phi(u,u';\eta)-\Phi(u,u';\eta') \|_E \leq C_\vep \|u'-u\|_H \|\eta -\eta'\|_E,  \label{bd2}
\end{eqnarray}
$\text{ for any }(\eta,\eta')\in  B_E(R_\eta)\times  B_E(R_\eta)$ and $(u,u')\in D_\delta$.

\smallskip

{\it Step 2: estimate \eqref{3.9}.} According to identity \eqref{taylor} and estimate
 \eqref{bd1},  for any $\eta \in \K$ and $\ u,u'\in D_\delta$ we have
\begin{eqnarray*}
\|S(u',\eta+\Phi(u,u';\eta))- S(u,\eta)\|_H & \leq& C \left(\|u'-u\|_H^2 +\|\Phi(u,u';\eta)\|_E^2\right) + \vep \|(D_u S)(u'-u)\|_V\\
&\leq & \|u'-u\|_H \left(C \vep + (1+C_\vep^2)C\|u'-u\|_H \right)\\
&\leq & \tfrac12 \|u'-u\|_H,
\end{eqnarray*}
where to get the last inequality we choose $\vep$ so small  that $C\vep<\frac14$  next choose  $\delta\le \delta(\vep) $ so small
 that $(1+C_\vep^2)C \delta <\frac14$. This proves  \eqref{3.9}.
 \smallskip

 \noindent
  {\it Step 3: estimate \eqref{3.8}.}
Recall that $ E_M=$span$\,\{e_1, \dots, e_M\}$.
Denote by  $E^\perp_M$ the complementary space of $ E_M$ in $E$. 
For any $ A \in \mathcal{B}_E $ we will denote by $A_M$
and $A_M^\perp$ its
 projections  to $ E_M$ and   $ E_M^\perp$, respectively.
Consider the Hilbert brick $\tilde\CK$, defined in \eqref{brick}. Then $\tilde\CK_M$ is a parallelepiped, 
$\tilde\CK_M^\perp$ is a Hilbert brick in $ E_M^\perp$,
$$
\tilde\CK_M = \{ \sum_{i=1}^M x_i e_i \Big| \ |x_i|\leq b_i, i=1,...,M  \},
$$
and $\tilde\CK = \tilde\CK_M\times \tilde\CK_M^\perp$.

By  assumption (H1) the measure $\ell$ is supported by  $\tilde\CK$ and
 $\ell = \ell_{M} \otimes \ell^\perp_{M}$ where  $\ell_M$ is supported by
$\tilde\CK_M$ and  $\ell_M^\perp$  --  by $\tilde\CK_M^\perp$. Besides, $\rho_M =\rho_M(dv)$, where
$ E_M=\{ v\}$ and $\rho_M$ is a Lipschitz function, supported by $\tilde\CK_M$.

 For any fixed $A\in \CB_E, $ denote by $h_A$ the indicator function of $A$.
 According to  Fubini's theorem,
$$
\ell(A)= \int h_A(e) \ell(de)= \int_{\tilde{\CK}^\perp_M} \ell^\perp_{M} (dw) \int_{\tilde\CK_M} h_A(w+v) \ell_M(dv)\\
= \int_{\tilde{\CK}^\perp_M}  \ell^\perp_{M} (dw) \int_{\tilde\CK_M}h_A(w+v) \rho_M(v)dv.\\
$$

Similar,
$$
\Psi_*(\ell)(A)= \int h_A(e+\Phi(u,u';e)) \ell(de)
= \int_{\tilde{\CK}^\perp_M} \ell^\perp_{M} (dw) \int_{ \tilde\CK_M}h_A(w+v+\Phi(u,u';w+v)) \rho_M(v)dv.
$$
Since $\| \ell- \Psi_* \ell\|_{var} = \sup_{A\in \CB_E}|\ell(A) - \Psi_*(\ell)A|$, then \eqref{3.8} would follow if we prove that
\begin{equation}\label{gbd}
 |\int_{\tilde\CK_M }h_A(w+v) \rho_M(v)dv- \int_{\tilde\CK_M}h_A(w+v+\Phi(u,u';w+v)) \rho_M(v)dv  | \leq C\|u-u'\|\quad \forall w,u,u',
\end{equation}
where $C$ does not depend on $A$ as well as on $w,u,u'$. To do this let us consider the mapping
$$
\Phi_w: \tilde\CK_M \ni v \mapsto \Phi(u,u'; v+w) \in  E_M.
$$
By \eqref{bd1} and \eqref{bd2} the map's norm and the 
 Lipschitz constant are  bounded by
 $C_\vep\| u-u'\|_H \le C_\vep \delta$. We will assume that $\delta$ is so small that
\be\label{delta}
C_\vep \delta\le 1/2
\ee
(where $\vep$ was fixed at  Step 2 of the proof). By the Kirszbraun theorem (see \cite{Fed}), for each fixed $u, u'$ and $w$,
$\Phi_w$ extends to  a mapping $\tilde\Phi_w:  E_M\to  E_M$ with the same Lipschitz constant.  Now consider the mapping
$$
\Xi:  E_M\to  E_M,\quad  v \mapsto \xi =v+ \tilde\Phi_w(v).
$$
Due to \eqref{delta} its Lipschitz constant is $\,\le 1+C_\vep \delta\le 3/2$. The inverse mapping  $\Xi^{-1}:  E_M\to  E_M$ exists and
$$
\text{Lip}\, (\Xi^{-1} - \text{id}) \le C'_\vep   \| u-u'\|.
$$

Let us consider the second integral in \eqref{gbd} and write it as an integral over $ E_M$ with an integrand which vanishes outside $\tilde\K_M$:
$$
 \int_{\tilde\CK_M}h_A(w+v+\Phi(u,u';w+v)) \rho_M(v)dv =
 \int_{ E_M}h_A(w+v+\Phi(u,u';w+v)) \rho_M(v)dv .
$$
Passing there  from the variable $v$ to $\xi$ we write the integral as
$$
\int_{ E_M} h_A(w+\xi) \rho_M(v(\xi)) \Big| \det \frac{\partial \xi}{\partial v}\Big|^{-1}\,d\xi=
\int_{\Xi(\tilde\K_M)} h_A(w+\xi) \rho_M(v(\xi)) \Big| \det \frac{\partial \xi}{\partial v}\Big|^{-1}\,d\xi
$$

(concerning Lipschitz changes of variables in integrals over $\R^N$ see e.g. Theorem~3.2.5 in
\cite{Fed}).
Writing the first integral in \eqref{gbd} as $ \int_{\K_M}  h_A(w+\xi) \rho_M(\xi) \,d\xi\ $ and using that the mapping
$v\mapsto \xi$  and its inverse both are $C_\vep \| u-u'\|$-close to the identity in the Lipschitz norm, we get
\eqref{gbd}. The main lemma is proved.
\end{proof}

\begin{rem}\label{r_2}
If, as in Remark \ref{r_1}, the mapping $S$ is defined only on a subdomain $Q$ of $H\times E$, containing $H\times \K$, and the 
assumption of Remark~\ref{r_new_norm} holds, then assumption \eqref{delta} should be strengthened to 
$
C_\eps \delta \le \min( 1/2, r(R_+)),
$
where $r$ is the positive functions from Remark~\ref{r_1}.  Then the extension $\tilde\Phi_w$ can be chosen to be bounded in norm
by $r(R_+)$, so the points $(u', \Psi)$ stay in the domain $Q$, where $S$ is well defined. 

\end{rem}


\section{Applications } \label{s_2}

In this section we apply the abstract theorems above to the two-dimensional Navier-Stokes system 
 on the torus $\mathbb{T}^2  = \R^2/(2\pi  )\Z\oplus (2\pi  )\Z$ and to well posed 
 quasilinear parabolic systems on $\T^d$,  perturbed  by 
  random forces,  and
  prove that  these systems are exponentially  mixing. To do that  we will pass from the corresponding PDE
  \eqref{1}  to a discrete time system \eqref{S_oper}, will show  that the latter 
satisfies  the assumptions (A1), (H1), (H2), (H3) and then will apply Theorem~\ref{t_final}.  We will do that in the situation when
the space $H$ is a Hilbert space of functions of $x$ and $E$ is the space $L_2([0,1]; H)$. Accordingly, the Hilbert base $\{e_j\}$ of $E$
will be of the form $\{h_r\otimes \phi_i\}$, where $\{\phi_i\}$ is a base of $H$ and $\{h_r\}$ --  a base of $L_2(0,1)$. To apply the abstract 
theorems to the 2d Navier-Stokes system the base $\{\phi_i\}$ may be arbitrary, while to apply them to the 
quasilinear parabolic systems its elements 
should be bounded functions. Below we restrict ourselves to the case when $\{\phi_i\}$ is the Haar base.

\subsection{Random Haar series (``colored noises")}\label{s_noise}

We will apply Theorem~\ref{t_final} to equations \eqref{1},  perturbed by random forces $\eta_t$ of the form 
$$
\eta_t= \sum_{k=1}^\infty 1_{[k-1,k)} (t) \eta_k(t-k+1 ),
$$
where $\{\eta_k(\tau), 0\le \tau<1 \}_k$ is a sequence of i.i.d. random variables in $L^2([0,1], H)$. 
 It suffices  to define the process $\eta_1 =\eta\mid_{[0,1)}$. To do this, let us denote 
  by $\{  h_{jl}: j\in \N, 0\leq l\leq 2^j-1\}$ the orthonormal Haar basis on $[0,1]$ :
\be\label{haar} 
h_{00}(t)=\left\{
\begin{array}{ll}
1 & 0\leq t <1 \\
0 & t<0 \text{ or } t\geq 1, \end{array}
\right.
\ \ \ h_{jl}(t)=\left\{
\begin{array}{ll}
0 & t<l2^{-j} \text{ or } t\geq (l+1) 2^{-j} ,\\
2^{\frac j2}& l2^{-j} \leq t < (l+\frac12) 2^{-j},\\
-2^{\frac j2} & (l+\frac12)2^{-j} \leq t < (l+1 ) 2^{-j}.
\end{array}
\right.
\ee
Then  $\{  h_{jl}\phi_i \}_{ijl}$ is an  orthonormal basis of  $E=L_2([0,1], H)$  (we recall that  $\{\phi_i \}_i$ is a Hilbert basis  of $H$). 

 Now we define the process $\eta_1 =\eta\mid_{[0,1)}$  in $H$ as 
\be\label{NSnoise1}
\eta_1(t)= \sum_{i=1}^\infty b_i \eta^i_t \phi_i, \quad b_i\ne0 \; \forall\,i. 
\ee
Here $\sum_i b_i^2 < \infty$ and for $i\ge 1$, $ \eta^i_t $ is the random process 
\be\label{NSnoise2}
\eta^i_t=  \sum_{j=0} c_j^i \sum_{l=0}^{2^j-1} \xi^i_{jl} h_{jl}, \quad  c_j^i\ne0\;\; \forall\, i,j, 
\ee
where
$$
\sum_{ijl}   (c^i_j b_i)^2 = \sum_i b_i^2 \sum_j 2^j (c^i_j)^2 =: R_\eta^2 <\infty,
$$
%

 and $\{  \xi^i_{jl}\}_{i,j,l}$ are independent real 
random variables with $|\xi^i_{jl} | \leq 1$, whose density functions $\rho_{ijl}$ are Lipschitz--continuous 
 and $\rho_{ijl}(0) \neq 0$. Denote the basis $\{  h_{jl}\phi_i \}_{i,j,l}$ as $\{ e_k , k\geq 1 \}$, where 
 $$
 j(k)\to\infty,\,\; i(k) \to\infty\quad \text{as} \; k\to\infty. 
 $$
 
Then assumption $(H1)$ obviously holds.

\subsection{The 2D Navier-Stokes system}\label{s_2dNS}

Consider the 2D Navier-Stokes system
\be\label{NS}
\partial_t u  - \nu \triangle u + \langle u, \nabla \rangle u  + \nabla p = \eta(t,x) , \ \ \ \text{div}\ u=0,\quad \int u\,dx=0, 
\ee
where $x\in \T^2= \R^2/ 2\pi \Z^2, u=(u_1,u_2)$, $p$ is an unknown scalar function (the pressure), $\nu>0$ is the 
viscosity  and $\eta$ is a random force as in the previous section. 
Denote by $\MH$ the $L_2$--space 
$$
  \MH= \Big\{ u \in L_2(\T^2, \R^2) :\ \text{div}\ u=0, \ \int_{\T^2} u(x)dx=0 \Big\}\ 
$$
with the norm  $\|\cdot \|$. For $m\in\N$ we set $ \MH^m:= \MH \cap H^m(\T^2, \R^2)$, 
where $H^m$ is the Sobolev space on the torus, equipped with the Sobolev norm $ \| \cdot \|_m$, and for  $m\in -\N$ denote
by  $ \MH^m$ the closure of  $ \MH$ in the Sobolev norm  $ \| \cdot \|_m$.  Let $\{\phi_j, j\geq 1\}$ be the usual trigonometric Hilbert basis of  $\MH,$ (e.g. see Chapter 2 in \cite{KS}). This  also is an orthogonal basis of  every  space $\MH^m.$ 
\\

Suppose  that $\eta(t, \cdot) \in \MH$ for all $t$, a.s. Applying the  Leray projection $\Pi: L_2(\T^2, \R^2) \rightarrow \MH $ to equation \eqref{NS} we obtain the following nonlocal PDE:
\be\label{NS2}
\partial_t u + \nu L u + B(u) = \eta(t,x), \ \ \ u(0)=u_0,
\ee
where $L=-\Pi\Delta $ and $B(u)= \Pi(\langle u, \nabla \rangle u )$. 

For $j\in \Z$ denote
$$
 E_j = L_2([0,1], \MH^j).
$$
By the classical results (see e.g. Chapter 2 in \cite{KS}), if $m\geq 1$, then for any  $u_0 \in \MH^{m-1}$ and $ \eta \in  E_{m-2}$ 
eq.~\eqref{NS} with $0\le t\le1$  has a unique solution 
$$
u \in \U_m:=\left\{ u_. \in E_m: \partial_t {u}_. \in E_{m-2} \right\}  \subseteq C([0,1];\MH^{m-1}).
$$
We equip the space $\U_m$ with the Hilbert 
norm $\|u\|_{\U_m}^2=\int_0^1\|u_t\|_{m}^2  dt  + \int_0^1\| \partial_t u_t\|^2_{m-2} dt $. 

Consider the dynamical system 
\be\label{sys1.1}
  S(u_{k-1},\eta_k)= u(1), 
\ee
where $u(t)$ is a solution of equation \eqref{NS2} with $u(0)= u_{k-1}$ and $\eta= \eta_k$. We are going to apply to
system  \eqref{sys1.1} Theorem \ref{t_final} with $H=\MH^1$ and $E= E_1$. Accordingly we rescale the basis $\{\phi_i\}$ to
be an orthonormal basis of the space $\MH^1$.
The validity of  assumption (H1) 
was already checked in  Section~\ref{s_noise}. In order to verify the the remaining assumptions we need two auxiliary results:

\begin{lem}\label{mainlem2} 
For any $u \in \U_2$ and $j=0$ or 1, 
consider the mapping 
\[
A_2(u):  \U_{j+2} \longmapsto   E_j \times \MH^{j+1} ; \ \  v   \longmapsto  \left( \partial_t v - \nu\Delta v + \Pi \left( \langle u, \nabla \rangle v + \langle v, \nabla \rangle u \right), v(0)  \right) .
\]
Then $A_2(u)$ is a linear isomorphism. Furthermore, the norm of its inverse  depends only on $\|u\|_{\U_2}.$
\end{lem}

The proof easily follows by Galerkin' method.

\begin{prop}\label{improreg}
For $u_0 \in \MH^1$ and  $\eta \in  E_1$, we have $S(u_0,\eta) \in \MH^2.$ Furthermore, the mapping  $S : \MH^1 \times  E_1 \mapsto \MH^2$ is analytic.

\end{prop}

\begin{proof}
Consider equation \eqref{NS2} with $u_0 \in \MH^1$ and  $\eta \in  E_0$. It has a unique solution $u \in \U_2.$ By Lemma \ref{mainlem2} with $j=0$ and the implicit function theorem, $u$  analytically depends on $(u_0,\eta)$ (see \cite{K82} for details), so the mapping 
\[\begin{array}{lc}
A_1: & \MH^1 \times  E_0  \longmapsto \U_2 ;  \ \ (u_0, \eta)  \longmapsto  u,
\end{array}
\]
is analytic, as well as the mapping 
$
 \MH^1 \times  E_1  \longmapsto \MH^1,  \ (u_0, \eta)  \mapsto  u_1. 
$
It remains to improve the regularity and show that the map 
\[\begin{array}{lc}
S: & \MH^1 \times  E_1 \longmapsto  \MH^2; \ \    (u_0, \eta)  \longmapsto  u_1,
\end{array}
\]
is analytic. Note that 
\be\label{**}
S(u_0,\eta)= \int_0^1 (d/dt)S(tu_0, t \eta)\,dt  = \int_0^1 (D_{u_0} S(tu_0, t \eta) u_0 + D_\eta S( tu_0,t\eta) \eta) dt,
\ee
so it suffices  to show that $D_{u_0} S(u_0,\eta)h $ and $ D_\eta S(u_0,\eta) \xi$ as mappings 
$\MH^1 \times  E_1\times \MH^1 \to \MH^2$ and $ \MH^1 \times  E_1\times  E_1\to \MH^2$, respectively, are analytic.\\
Denote by $u=u(u_0, \eta)
\in \U_2$ a solution of \eqref{NS2} with $(u_0, \eta) \in \MH^1 \times  E_1.$ Then $D_\eta S(u_0, \eta) \xi= v_1(0 , \xi), \xi \in  E_1$ where $v_t(v_0, \xi)$ stands for  a solution of the following  linear equation
\be\label{linear3}
\partial_t v + \nu L v + \Pi \left( \langle u, \nabla \rangle v + \langle v, \nabla \rangle u \right)= \xi, \ \ v(0)=v_0.
\ee

{\it The map $D_\eta S$:} 
By Lemma \ref{mainlem2} with $j=1$, we have $v_.(0,\xi) \in \U_3.$  It remains to show $v_.(0,\xi)$ analytically depends on $u,$
 which would imply that $D_\eta S( u_0, \eta) \xi= v_1$ is an analytic mapping. To this end consider the mapping
$$
F :  \U_2 \times \U_3  \longrightarrow   E_1 \times \MH^2, \quad 
  (u,v)   \longrightarrow  A_2(u)v =\left( \partial_t v - \nu \Delta v + \Pi \left( \langle u, \nabla \rangle v + \langle v, \nabla \rangle u \right), v(0)  \right) .
$$
To prove the required analyticity it suffices to apply Lemma \ref{mainlem2}  and the implicit function theorem to  the equation 
$F(u,v)=(\xi,0).$

{\it The map $D_{u_0} S$:}
Similarly  for any $h\in \MH^1$, $D_{u_0} S(u_0, \eta) h= \bar{v}_1,$ where $\bar{v}_t$ solves \eqref{linear3} with $\xi=0,  \ \bar{v}_0=h .$ By Lemma~\ref{mainlem2} with $j=0,$ $\bar{v}\in \U_2$; so $ \bar v(1) \in \MH^1$. To improve the smoothness and show that 
$ \bar v(1) \in \MH^2$ we consider the function $w=t\bar{v}.$ Calculating $\partial_t w$ and using that $\bar{v}$ satisfies \eqref{linear3} with $\xi=0,$ we find that in its turn, $w$ satisfies \eqref{linear3} with $v_0=0,  \ \xi=\bar{v}\in \U_2$. Then by Lemma  \ref{mainlem2}, $w\in \U_3.$ It follows that $\bar{v}_1=w_1 \in \MH^2$ analytically depends on $(u_0, \eta)$ for the same  reason as above. This proves the required analyticity of the mapping $S$. 
\end{proof}

The last proposition implies (A1). To verify (H2)  we note that since
\begin{eqnarray*}
\frac12 \frac{d}{dt} \| u \|_1^2 = \lan L u, \partial_t u \ran =
- \nu \|u\|_2^2 + \lan L u , \eta \ran
 \leq 
  -\frac{\nu}{2} \|u\|_2^2 + \frac{1}{2\nu} \|\eta \|^2,
\end{eqnarray*}
then
$
\| u_1\|_1^2 \le e^{-\nu} \| u_0\|_1^2 + \beta
$
 for  $\eta\in\K$, with 
 some constant $\beta$. So (H2) follows.

It remains to check  assumption (H3):

\begin{lem}\label{DI-NS}
For any 
$(u_0, \eta) \in \MH^1\times  E_1$ the  mapping $D_\eta S (u_0,\eta): E_1\mapsto \MH^1$ has dense image. 
\end{lem}

\begin{proof}
Denote by 
$u =u(u_0, \eta) 
\in \U_2 $ the solution of equation \eqref{NS2} and by $v_t(v_0, \xi)$ the solution of equation \eqref{linear3}. Define $\mathcal{L}: E_1\mapsto \MH^1$ by $\ML(\xi)= v_1(0,\xi).$ Then $D_\eta S(u,\eta) \xi= \ML(\xi),$ so we should  show that the mapping 
$\ML: E_1\mapsto \MH^1$ has dense image. According to Fredholm's alternative, we only need to verify  that the adjoint operator 
$\ML^*: \MH^1 \mapsto  E_1$ has trivial kernel.

In order to do that let us consider the dual problem
\be\label{linear2}
\left\{
\begin{array}{ll}
- \partial_t w + \nu Lw - \Pi (\langle u, \nabla \rangle w) +  \Pi (du)^* w=0,&\\
w_1=w,&
\end{array}
\right.
\ee
where 
$\ 
( (du)^* w)^j = \sum_{l=1}^2 \frac{\partial u^l}{\partial x^j} w^l.
$
It is dual to the problem \eqref{linear3} with $\xi=0$ in the sense that if $v$ solves  \eqref{linear3} with $\xi=0$ and
$w$ solves  \eqref{linear2}, then 
$$
\langle v_t, w_t\rangle = \text{const}. 
$$
For a fixed  $\tau \in[0,1],$ let $S_\tau^t: v \mapsto {v}_t$, $\tau \le t\le1$, be 
 the resolving operator for equation \eqref{linear3} with 
 $\xi=0$ and the initial value ${v}_\tau=v$. Similar, let   $\tS_1^t: w \mapsto w_t$ be the resolving operator for eq.~\eqref{linear2} with terminal value $w_1=w.$ Then $\tS_1^t$ is   the dual operator for $S_t^1$  with respect to the   $L_2$--scalar product. 
 Accordingly, for any  $\eta \in  E_1$ and $w\in \MH^2,$ by Duhamel's principle,   
\begin{eqnarray} \label{relation1} 
\int_0^1 \langle \eta_t, (\ML^* w)_t  \rangle_1 dt = \langle \eta, \ML^* w \rangle_{ E_1} = \langle \ML(\eta), w  \rangle _{1} 
= \langle \int_0^1 S_t^1 \eta_t  dt, w \rangle_1 =\int_0^1 \langle \eta, \tS_1^t  L w \rangle dt.
\end{eqnarray}
If  $w\in \MH^1$, then we approximate $w$ in $\MH^1$ by 
smooth functions $w^n, n\ge1$, and substitute 
\be\label{xi_n}
\tS_1^t  Lw^n = L^{1/2} \xi_t^n. 
\ee
Then $\xi_1^n = L^{-1/2} L w^n = L^{1/2}  w^n$ and 
$$
- \partial_t \xi^n + \nu L\xi^n +L^{-1/2} \Pi \big(-\langle u, \nabla \rangle L^{1/2} \xi^n  +  (du)^* L^{1/2}  \xi^n\big)=0,\quad
\xi_1^n = L^{1/2} w^n.
$$
Taking a scalar product of this equation with $\xi^n$ in $\MH$ and using that  the vectors 
$\xi_1^n =L^{1/2}  w^n$ are bounded uniformly in $\MH$ we find that 
$
| \xi^n|_{\U_1} \le C(| u|_{\U_2} , \| w\|_1  ),
$
uniformly in $n$. 
So $\xi^{n_j} \rightharpoonup \xi \in\U_1$ weakly in $\U_1$, for a suitable sequence $n_j\to\infty$. 
 From this convergence, \eqref{xi_n} and \eqref{relation1} 
we get that 
$$
(L\ML^* w)(t) =L^{1/2} \xi_t \in C([0,1], \MH^{-1}), \quad  (L\ML^* w)(1) =L^{1/2} \xi_1 =  Lw. 
$$
Now let $\bar w_1 \in \MH^1$ be such that $\ML^* \bar w_1 =0$. Then $L\bar w_1=(L\ML^* \bar w_1)(1)=0,$ which implies $\bar w_1=0$ and completes the proof.
\end{proof}

It follows that Theorem \ref{t_final} applies and implies that 

\begin{thm}\label{t_2dNSE}
Suppose that the 
noise $\eta$ has the form \eqref{NSnoise1}, \eqref{NSnoise2}, where $\{\phi_i\}$ is the trigonometric  Hilbert basis of $\MH^1$. 
Then for any $\nu>0,$ the Markov chain 
 $(u_k , P_k)$, defined  by the 2D Navier-Stokes equation \eqref{NS} in $\MH^1$, 
 has a unique stationary measure $\mu_\nu\in \MP(\MH^1)$. Furthermore, for any $R>0$ and  any measure $\mu$ supported by the  ball $B_{\MH^1}(R)$, there exist $C= C(R)>0$ and 
 $\kappa = \kappa(R) \in(0,1)$ such that
$$
\| \MP_k^*\mu - \mu_\nu \|_{L(\MH^1)}^* \le C\kappa^k, \quad \forall\, k\ge1.
$$
\end{thm}
\medskip

Let $u^\nu_t$ be a solution of \eqref{NS} such that $\CD(u^\nu_0)= \mu_\nu$. 
Then  $\CD(u^\nu_k)= \mu_\nu$ for $k\in\N$, while
for $t=k+\tau$, $0\le\tau<1$, we have
$$
\CD(u^\nu_t) = \Sigma ^\tau \CD(u^\nu_k) = \Sigma^\tau\mu, 
$$
where $  \Sigma^\tau$ is a Lipschitz operator in the space $( \MP(X), \| \cdot \|_L^*)$. So we get 

\begin{coro} \label{c_2dNSE}
Under the assumptions of Theorem \ref{t_2dNSE} let $u_t(x)$ be a solutions of \eqref{NS} such that $\CD (u_0)= \mu$. 
Then
$$
\| \CD(u_t)  -   \CD(u^\nu_t) \|_{L(\MH^1)}^* \le C\kappa^t, \quad \forall\, t\ge0. 
$$
\end{coro}

Note that the mapping $S: \MH^0 \times E_1 \rightarrow \MH^1$ is also well defined (see e.g. \cite{KS}), and for any initial distribution $\mu \in \MP(\MH^0),$ $\CD(u_1) \in \MP(\MH^1)$. By applying Theorem \ref{t_2dNSE} and Theorem \ref{t_final}, we have the following result.

\begin{coro} \label{c2_2dNSE}
For any $\mu \in \MP(\MH^0),$ supported by the  ball $B_{\MH^0}(R)$, there exist $C= C(R)>0$ and 
 $\kappa = \kappa(R) \in(0,1)$ such that
$$
\| \MP_k^*\mu - \mu_\nu \|_{L(\MH^1)}^* \le C\kappa^k, \quad \forall\, k\ge1.
$$
Furthermore, if $\nu \in \MP(\MH^0)$ is any measure, then
$
 \MP_k^* \nu \strela \mu_\nu$  as $k \rightarrow \infty.
$
\end{coro}

\begin{rem}\label{e_kicks}
 Now consider the 2D NSE \eqref{NS}, where $\eta(t,x)$ is a random kick force (see \cite{KS}, Section~2.3):
 $
 \eta= \sum_{k=1}^\infty \eta_k\,\delta(t-k).
 $
 Here $\{\eta_k\}$ is a sequence of i.i.d. random variables in $\MH^2$.  Solutions $u_t, t\ge0$, of the equation
 define the map $S:\MH^1 \times\MH^2 \to \MH^1$ by the relation
 $$
 S(u_0, \eta_1) = u_1 = S_1(u_0)+\eta_1,
 $$
 where $S_1$ is the time-one flow-map for the free equation \eqref{NS}${}_{\eta=0}$. According to Proposition \ref{improreg}, the map $S$ satisfies the regularity assumption 
 (A1) with $E:=\MH^2$. Now as in \cite{KS}, Section~3.2.4, set $\eta_k = \sum_{j=1}^\infty b_j \xi_{jk} \phi_j$, 
 where $\{b_j, j\ge1\}$, is an $l_2$--sequence  of nonzero real numbers, $\{\phi_j, j\ge1\}$, is the orthonormal   trigonometric basis of the space $\MH^2$,
  and $\{\xi_{jk}\}$ are i.i.d. real variables whose density function is Lipschitz--continuous 
 and does not vanish at the origin. Then (H1),  (H2)  and (H3) hold trivially. 
  Applying Theorem~\ref{t_final} we recover the well known result that the 2D~NSE with a non-degenerate random kick force is
 exponentially mixing, see in \cite{KS}. 

\end{rem}
Similarly Theorem~\ref{t_final} applies to the CGL equation \eqref{cgl} as below in Example \ref{e_cgl}, where $\eta$ is a nondegenerate kick-force.


\subsection{Quasilinear parabolic systems on $\T^d$}

In this part we consider quasilinear parabolic systems 
\be\label{qps}
\partial_t u= \Delta u + f(x, u, \nabla u )+ \eta, \ \ \  u(0)=u_0,
\ee  
where $f(x,0,0) =0$, 
$u=u_t(x) \in \R^n$, $x \in \T^d$ and 
 $f : \T^d \times \R^n \times \R^{n\times d} \mapsto \R^n$ is a $C^\infty$--smooth function. We restrict ourselves to the  case when 
 the solution $u$ is sought in the space 
 $$
 H^m= H^m(\T^d; \R^n), \quad m > \frac d 2 +1,
 $$
 and $\eta_t$ is a process as in Section \ref{s_noise} with $H=H^m$.

  We are going to apply Theorem~\ref{t_final} with $H=H^m$. To do that we will  make  two assumptions concerning  the well-posedness and regularity of  eq.~\eqref{qps}. To formulate them we define  the following spaces, where $k\in\Z$ and $T>0$:
  $$
  E_k^T =L_2([0,T]; H^k), \qquad 
  \U_k^T =\{u\in E_k^T: \partial_t u  \in E^T_{k-2}\} \subset C([0, T]), H^{k-1}),
  $$
  and abbreviate $E_k^1 = E_k, \; \U_k^1 = \U_k$.

\begin{itemize}
\item[(A0)](\textbf{Well-posedness}). For any $M, T>0$, $u_0\in H^m$ and any $\eta\in E^T_{m-1}$ such that 
$\| \eta_t\|_{m-1}  \le M$ for all $t$, the problem \eqref{qps} has a unique solution $u\in \U^T_{m+1}$. It satisfies 
\be\label{R_+}
\| u_t\|_m \le C(M, \| u_0\|_m)\quad \forall\, t. 
\ee


\item[(H$1'$)](\textbf{Random force}). The force $\eta$ has the form \eqref{NSnoise1},  \eqref{NSnoise2}, where $\{\phi_i\}$ is a Hilbert
base of $H^m$, $\{h_{ij}\}$ is the Haar base and 
$$
M^2= (\sup_i \sum_j 2^{\frac j 2} |c_j^i|)^2 (\sum_i b_i^2) <\infty. 
$$

\item[(H$2'$)](\textbf{Dissipativity}). There exists $0\leq r' \leq m$ and $\kappa>0,$ such that if $u$ solves \eqref{qps} with $\eta=0$ and $u_0\in H^m$,  then 
\be
\|u_t \|_{r'}  \leq  e^{-\kappa t}\|u_0\|_{r'}.
\ee
\end{itemize}
\medskip

Note that due to (H$1'$) 
\be\label{eta_m}
 \| \eta_t\|_m \le M \quad \forall\, t,\; \forall\, \omega. 
 \ee
 
\begin{rem} 
Here we assume that the random force $\eta$ is bounded uniformly in $\omega$ and $t$, while in Section~\ref{s_2dNS} we
 assumed that it is bounded in $t$ in the $L_2$--sense (i.e. in the norm of the space $E$). This is needed since the class of equations 
 \eqref{qps} for which it is possible to prove well-posedness of the initial value problem for  $u_0\in H$ and  $\eta\in E^T_{m}$ 
  is smaller than the class of equations for which we can prove the well-posedness for the problem with bounded 
 in time $\eta$ (e.g. we cannot prove that assumption (A0) holds for the CGL equations as in 
 Example~\ref{e_cgl} without the additional restriction   $\| \eta_t\|_{m-1}  \le M$). So working with bounded in time random 
 forces as in  \eqref{eta_m} we can apply our main theorem to a larger class of quasilinear equation  \eqref{qps}. 
\end{rem}

\begin{lem}\label{b} Under the assumption (A0)  there exists an open \nbh\ $Q$ of $H^m \times \K$ in $H^m\times E_{m-1}$ which
contains each point $(u_0, \eta) \in H^m\times \K$ with its vicinity  in $H^m\times E_{m-1}$ of radius which depends only on $\| u\|_m$, 
such that for every $(u_0, \eta) \in Q$ the
problem \eqref{qps} has a unique solution $u$. Moreover, the mapping
$$
Q\ni (u_0, \eta) \mapsto u\in \U_{m+1}
$$
is $C^\infty$--smooth, and for any $k\ge1$ 
 its $C^k$--norm is bounded on bounded sets. 

\end{lem}

\begin{proof}
Consider the mapping 
$$
\Phi: \U_{m+1} \mapsto H^m \times E_{m-1}; \ \ \ u \mapsto \left(u_0, \partial_t u- \Delta u - f(x, u, \nabla u ) \right).
$$
It is $C^\infty$ smooth (the smoothness of the nonlinear component follows from a much more general result in \cite{RS}, pp.~14,~381),
and by (A0), its image contains $H^m \times \K.$ For  any $(u_0, \eta) \in H^m\times \K$ and for $u=\Phi^{-1}(u_0, \eta) $, the linear mapping
$$
d\Phi(u) : \U_{m+1} \mapsto H^m\times E_{m-1}  ; \ \  d\Phi(u) (v) = (v_0,  \Delta v + D_u f(x,u,\nabla u) v + D_{\nabla u}f(x,u, \nabla u) \nabla v 
).
$$
is an isomorpism of $ \U_{m+1}$ and $H^m\times E_{m-1}$ (see e.g. in \cite{ASU68}), so by the inverse map theorem the point  $(u_0, \eta) $ has a
\nbh\ $Q_{ (u_0, \eta) } \subset H^m\times E_{m-1}$, where the inverse mapping $\Phi^{-1}$ exists and is $C^\infty$. By the constructive 
nature of the inverse map theorem, the size of the \nbh\ and the norms of derivatives of the inverse mapping are bounded in terms 
of the norms $\| u_0\|_m$ and $|\eta|_{E_{m-1}}$. Since $\eta$ belongs to the compact set $\K$, these quantities may be chosen
  $\eta$--independent.
 Taking for $Q$ the open set $\cup_{ (u_0, \eta)} Q_{ (u_0, \eta)}$ we arrive at
 the conclusion.
\end{proof}

 Due to the lemma and \eqref{eta_m}, 
  Remarks \ref{r_1}, \ref{r_2} apply to the equations \eqref{qps} if the assumptions $(A0)-(H2')$ hold.

\begin{example}\label{e_cgl} 
 Consider the complex Ginzburg-Landau (CGL)  equation:
\be\label{cgl}
\partial_t u -(\nu_1+ i\nu_2) \Delta u  + \gamma u + i  |u|^{2r} u = \eta, \quad x\in \T^d, \quad 
 u(0)=u_0, 
\ee
where $\gamma>0, \nu_1>0$ and $\nu_2\ge0$. This is an example of system \eqref{qps}
with $n=2$. Assumption (H$2'$) with ${r'}=0$ obviously holds for all equations \eqref{cgl}. Assumption (A0) also is
fulfilled for a large class of the equations. In particular, it holds if $\nu_2>0$ and  $d\le2$, $r\in \N$ or 
$d=3$ and $r=1$; or if $\nu_2=0$ and  $r\in \N$. See Appendix. 
\end{example}
\medskip

We may assume that 
$$
u_0 \in B_{H^m}(R)\quad \text{a.s.},
$$
for some $R>0$. 
According to assumption (A0), the mapping 
$S: H^m \times  E_{m-1} \supset Q_R \to H^m$ is $C^2$--smooth and its $C^2$--norm is bounded. 
 Here $Q_R$ is a \nbh\ of 
$ B_{H^m}(R) \times \K$ in $ H^m \times  E_{m-1}$ which contains each point of $ B_{H^m}(R) \times \K$ with its \nbh\ 
of radius $c_R>0$. As above, 
$S(u_0, \eta) =u_1$ where
$u_t$, $0\le t\le1$, is a solution of \eqref{qps}. So \eqref{qps} defines in $H^{m}$ the 
random dynamical system  \eqref{sys1.1}. By \eqref{eta_m} and 
(A0) the trajectories of \eqref{sys1.1} with $u_0 \in B_{H^m}(R)$
satisfy $\| u_k\|_m \le R_+$, $R_+=R_+(R,M)$,  for all $k$, a.s. 
 In the following three steps we will  check that  the 	assumptions 
 (A1),  (H2) and (H3) hold with $H=H^m, E=E_{m}$ and $V= H^{m+1}.$  Verifying assumption (H2), we only need to take into account Remark~\ref{r_new_norm}.

\bigskip

{\textbf{Step 1.}} ({\textbf{Improvement of regularity}) 
In view of (A0) to  prove (A1) it  remains to show that 
 $S$ restricted to $Q_R\subset  H^m \times E_m $ takes values in $H^{m+1}$ and is smooth. Evoking  equality \eqref{**} (where now
 $S$ is the operator in \eqref{sys1.1}) we see that it suffices to examine 
  the regularity of the maps 
   $D_{u_0} S(u_0,  \eta) h$ and $D_\eta S( u_0, \eta) \xi.$ Let $u = u(u_0, \eta) \in \U_{m+1}$  solves \eqref{qps} with $u_0 \in H^m$ and $ \eta \in E_m.$
Since  $D_\eta S(u_0, \eta) \xi = v_1(0,\xi;u)$ and $D_{u_0} S(u_0, \eta) h= v_1(h,0;u),$ where $v_t(v_0, \xi;u)$ solves the 
linear system
\be\label{linearqp}
\partial_t v= \Delta v + D_u f(x,u,\nabla u) v + D_{\nabla u}f(x,u, \nabla u) \nabla v + \xi, \ \ v(0)=v_0,
\ee
we only need to examine  the regularity of the mapping 
 $(v_0, \xi,u) \mapsto v_1(v_0, \xi;u).$ Due to the classical results from the linear parabolic theory (see   \cite{ASU68},~Part~7), if
 $u\in \U_{m+1}$, then for $j=m$ and $j=m+1$, for $v_0\in H^j, \xi\in E_{j-1}$ eq.~\eqref{linearqp}  has a unique solution 
 $v\in \U_{j+1}$. By an argument similar to that 
  in Proposition~\ref{improreg}, we  show that $v$ smoothly depends on $u \in \U_{m+1} $. So $v_1(0, \xi;u) \in H^{m+1}$ smoothly 
  depends on $u\in \U_{m+1}$, $\xi\in E_m$. To analyse $v_t = v_t(h,0;u)$, as in Section~\ref{s_2dNS} we consider
  $w= tv$. This vector-function satisfies \eqref{linearqp} with $v_0=0, \xi=v(h,0; u)$. So $v_1=w_1 \in H^{m+1}$ smoothly depends on
  $\xi\in \U_{m+1}$, i.e. on $h\in H^m$ and $u\in \U_{m+1}$. This implies the smoothness, required in (A1). 
  
\bigskip

{\textbf{Step 2.}}  Below we follow \cite{KNS}, Section 4.2. 
Firstly we claim that for any $ R>0$ and 
$u_0 \in H^m$ with $\|u_0\|_m < R$,
\be\label{a2}
\| S(u_0,0) \|_m \leq C_R \|u_0 \|_{r'},
\ee
where $C_R>0$ depends only on $R.$ Indeed, since $f(x,0,0)=0$, then 
$$
S(u_0,0)=  \int_0^1 (\partial/\partial t) S(tu_0,0 ) u_0\ dt
= \int_0^1 D_u S(tu_0,0 ) u_0\ dt, 
$$
and  we only need to show that 
 $\| D_u S(tu_0,0) u_0   \|_m \leq C_R \|u_0 \|_{r'}$ if $\|u_0\|_m\leq R$. 
Denote by  $\tilde{u}_t(tu_0, 0)$ the solution of equation \eqref{qps} with $u_0$ replaced by $tu_0$ and $ \eta=0.$ 
By  \eqref{eta_m} and \eqref{R_+}, $\| \tilde u_t\|_m \le C(R,r)=: R_+$ for all $t$. 
 Consider the following linear system 
\be\label{linearqp2}
\partial_t v= \Delta v + D_u f(x,\tu,\nabla \tu) v + D_{\nabla u}f(x,\tu, \nabla \tu) \nabla v ,  \ \ \ v(0)=u_0;
\ee
then $D_u S(tu_0,0) u_0   =v_1$. Choosing any points
$$
0<t_{r'+1} <t_{r'+2} <\dots  <t_{m} =1
$$
and arguing as at Step 1 we find that
$\ 
\| v_{t_{r'+1}}\|_{r'+1} \le C'_R \| u_0\|_{r'} , \dots, \| v_{t_{m}}\|_{m}  \le \| v_{t_{m-1}}\|_{m-1}. 
$
So $\| v_1\|_m \le C_R^{''} \| u_0\|_{r'}$ 	and \eqref{a2} follows.

 Due to \eqref{a2} and Assumption (H$2'$), in the space $H^m$ 
 there exists an equivalent norm $\| \cdot\|'_m$  such that for any $u_0 \in H^m$ with $\| u_0\|_m \leq R,$
\be\label{S_q} 
\| S(u_0,  0)\|'_{m} \leq q \|u_0\|'_m , 
\ee
where $q\in (0,1)$ depends on $R.$ Indeed, let $\| \cdot \|'_{m}:= \|\cdot \|_{r'} + \vep \| \cdot \|_m,$ where $\vep>0$ is a parameter to be defined. We have 
$$
\| S(u_0,  0)\|'_{m} = \| S(u_0,  0)\|_{r'} + \vep \| S(u_0,  0)\|_{m} \leq e^{-\kappa} \|u_0\|_{r'} + \vep C_R\|u_0\|_{r'} \leq (e^{-\kappa}+\vep C_R) \|u_0\|'_m.
$$
It remains  to choose $\vep$ so small that $q:= e^{-\kappa}+\vep C_R<1.$ Since the $C^2$--norm of the map 
$S:Q_R \to H^m$ is bounded, then \eqref{S_q} implies \eqref{linearbound} for $\| u\|_m \le R_+$ with the norm $\|\cdot \|_m$ 
replaced by  $\|\cdot \|'_m$. This proves  assumption   (H2) in the  weaker form, suggested in Remark~\ref{r_new_norm}.

\bigskip

{\textbf{Step 3.}} Now we verify (H3), i.e. check that 
 for any $u_0 \in H^m, \eta \in E_m,$ the linear operator $D_\eta S(u_0, \eta): E_m \mapsto H^{m}$ has dense image. Define the
  operator $\ML: E_m \mapsto H^{m} $ by $\ML(\xi)= v_1(0, \xi),$ where $v_t(v_0, \xi)$ solves  equation \eqref{linearqp}. By Fredholm's alternative it suffices  to show that  $\ML^*: H^{m} \mapsto E_m$ has trivial kernel.
 Denote $S_t^1: v_t \mapsto v_1$ the resolving operator for equation \eqref{linearqp} with $\xi=0$, $v(t)=v_t$.
 Consider the adjoint system 
 $$
-\partial_t w = \Delta w + \big(D_u f(x, u, \nabla u)\big)^* w - \text{div}\, \big[ \big( D_{\nabla u}f(x,u, \nabla u)\big)^*   w\big], 
 \ \ \ w(1)=w_1,
$$
and denote by $\tS_1^t : w_1 \mapsto w_t$, $0\le t\le1$, its resolving operator with initial condition at $t=1$.
Then  $\langle v_t , w_t \ran\equiv \text{constant},$ so the operator 
$\tS_1^t$ is the $L^2-$dual  of $S_t^1.$ It follows that, for  $\eta \in E_m$ and $ w \in H^{2m},$ 
\begin{eqnarray*}
&&\int_0^1 \langle \eta_t, (\ML^* w)_t  \rangle_m dt = \langle \eta, \ML^* w \rangle_{E^1_m} = \langle \ML(\eta), w  \rangle _{m} = \langle \int_0^1 S_t^1 \eta_t  dt, w \rangle_{m}\\
& = & \int_0^1 \langle L^{m} w, S_t^1 \eta_t  \rangle dt = \int_0^1   \langle \tS_1^t L^{m} w, 
  \eta_t  \rangle dt ,\quad L=1-\Delta.
\end{eqnarray*}
 If $w\in H^m$, then arguing as in Section~\ref{s_2dNS} 
we get that 
$$
L^m \ML^* w = L^{m/2} \xi, \; \;\xi\in \U_m  ;  \qquad \xi_1 = L^{m/2} w. 
$$
So $\ML^* w=0$ implies that $\xi_1=0$ and $w=0$, i.e.  $\ML^*$ has trivial kernel.
\\

Now an application of  Theorem \ref{t_final} implies the validity for  eq. \eqref{qps}  of obvious reformulations of 
Theorem~\ref{t_2dNSE} and Corollary~\ref{c_2dNSE}.

\subsection{Appendix.}

{\it Case $\nu_2>0$.}  We will only discuss equations with $d\ge2$ and to  simplify notation take $\nu_1=\nu_2=\gamma=1$:
\be\label{x1}
\partial_t u -i \Delta u+(u-\Delta u) +i |u|^{2r} u = \eta_t, \quad \| u_0\|_m =:R,
\ee
where $\|\eta_t\|_{m-1} \le M$ for all $\eta$, with a fixed $m>d/2$. We start with apriori estimates, assuming that $u$ is a smooth solution of 
\eqref{x1}. 
\smallskip

{\it Step 1}.  Noting that eq. \eqref{x1} with $\eta=0$ and with removed term $u-\Delta u$ is hamiltonian with the Hamiltonian 
$
H(u) = \int\big( \tfrac12 | \nabla u|^2 + \tfrac1{2r+2} |u|^{2r+2}\big) dx
$
and since  $\nabla H =-\Delta u +|u|^{2r} u$, we find that  
\begin{equation*}
\begin{split}
\partial_t &\big(H(u_t) +\tfrac12 \|u_t\|_0^2\big) 
 = \langle \nabla H +u, \Delta u-u\rangle +  \langle \nabla H +u, \eta \rangle\\  &\le -C(H(u_t) + \tfrac12 \| u_t\|_0^2) + \|u_t\|_1 \|\eta_t\|_1 
+\int |u|^{2r+1} |\eta| \,dx 
\le - \tfrac{C}{2} (H(u_t)+ \tfrac12 \| u_t\|_0^2)  +C_1( \|\eta_t\|_1). 
\end{split}
\end{equation*}
So
\be\label{x2}
H(u_t)  +\tfrac12 \|u_t\|_0^2 \le C(R,M)\quad \forall\, t\ge0. 
\ee
\smallskip

{\it Step 2}.  Now let us consider $\tfrac12\partial_t \| u\|_m^2$, where $  \| u\|_m^2 = \|u\|_0^2 + \|\nabla^m u\|_0^2$. 
We have:
\be\label{x33}
\tfrac12\partial_t \| u\|_m^2 + \| u\|_{m+1}^2 \le | \langle \nabla^m |u|^{2m} u, \nabla^mu\rangle |
+ |\langle \eta_t,( (-\Delta)^m +1) u \rangle |.
\ee
The first term in the r.h.s. is bounded by a finite sum of the terms of the form
$$
U = C\int |v^1\dots v^{2r+2} |\,dx, \quad v^j = \partial^{a_j} u\; \text{or} \; v^j=\partial^{a_j} \bar u,
$$
where 
\be\label{x4}
|a_1|+ \dots+ |a_{2r+2}| =2m,\quad |a_j |\le m \; \forall\, j. 
\ee
Then
\be\label{x5}
U \le \prod_1^{2r+2} |v^j|_{p_j}, \quad \sum \frac1{p_j} =1. 
\ee
By \eqref{x2}  and the Sobolev embedding, 
$$
|u_t|_{L_q}  \le C_1(R, M, \eps) \;\; \forall\, t\ge0, \qquad \frac1{q} = \frac{d-2+\eps}{2d}.
$$
Here and below $\eps=0$ if $d\ge3$ and  $\eps$ is any positive number if $d=2$. So by the Gagliardo--Nirenberg inequality 
\be\label{x6}
|v^j_t|_{p_j} \le C \| u_t\|_{m+1}^{\alpha_j} |u_t|_{L_q}^{1-\alpha_j} \le C_1(R, M, \eps)  \| u_t\|_{m+1}^{\alpha_j},
\ee
where
$$
\frac1{p_j} = \frac{|a_j|}{d} + \left(\frac12 - \frac{m+1}{d}\right) \alpha_j + \frac{(1-\alpha_j)(d-2+\eps)}{2d} =
\frac{2|a_j|+(d-2+\eps)}{2d} - \frac{\alpha_j(m-\eps)}{d} . 
$$
Denoting $\alpha = \sum \alpha_j$ we get from the last relation and \eqref{x4}  that 
$$
1= \sum\frac1{p_j} = \frac{4m+(d-2)(2r+2)+\eps'}{2d} - \frac{\alpha(m-\eps')}{d} . 
$$
So 
$\ 
\alpha =\big({4m+(d-2)(2r+2)-2d+\eps'}\big) /{2m}, 
$
which is  bounded by $2$ if $\eps'$ is sufficiently small and 
$
d< 2\frac{r+1}{r} ;
$
i.e. $r$ is any if $d=2$ and $r=1$ if $d=3$. 
Under this condition relations 
\eqref{x33},  \eqref{x5} and  \eqref{x6} jointly imply that 
$$
\partial_t \| u\|_m^2 + 2 \| u\|_{m+1}^2  \le C(M,R) \| u\|^\alpha_{m+1} +CM \|u\|_{m+1}, \quad \alpha<2. 
$$
From here 
$$
\| u_t\|_m \le C_1(M,R) \quad \forall \, t\ge0; \qquad | u|_{\U^T_{m+1}} \le C_2(M,R,T)\quad \forall\, T>0. 
$$

{\it Step 3}. The obtained estimates imply (A0) via Galerkin's method. 
\medskip

{\it Case $\nu_2=0$, $d$ is any.} In this case the function $r(t,x) = | u(t,x)|$ satisfies a differential inequality with the maximum
principle, e.g. see in \cite{K99}, where the white in time stochastic force $\zeta_t$ has to be replaced by the easier for 
analysis force $\eta_t$. This implies that 
$
|u_t|_{L_\infty} \le C(R,M)
$
for all $t$ and any dimension $d$. Then (A0) follows by the same argument as  at the  Steps~2-3 above.



\end{document}